\documentclass[5p,twocolumn]{elsarticle}

\usepackage{mathptmx}       
\usepackage{helvet}         
\usepackage{courier}        
\usepackage{type1cm}        
\usepackage{makeidx}         
\usepackage{graphicx}        
\usepackage{multicol}        
\usepackage[bottom]{footmisc}

\usepackage{amsmath}
\usepackage{bbm}
\usepackage{amssymb}
\usepackage{amsfonts}
\usepackage{cancel}
\usepackage{float}
\usepackage{color,soul}
\usepackage[normalem]{ulem}
\usepackage{array, setspace}
\usepackage{cancel}

\usepackage{multirow}
\usepackage{epstopdf}
\usepackage{amsthm}

\def\L2{{\cal L}_2}
\def\L2e{{\cal L}_{2e}}

\def\rea{\mathbb{R}}
\def\int{\mathbb{Z}}

\def\diag{\mbox{diag}}

\def\col{\mbox{col}}

\def\begequarr{\begin{eqnarray}}
\def\endequarr{\end{eqnarray}}
\def\begequarrs{\begin{eqnarray*}}
\def\endequarrs{\end{eqnarray*}}
\def\begarr{\begin{array}}
\def\endarr{\end{array}}
\def\begequ{\begin{equation}}
\def\endequ{\end{equation}}
\def\lab{\label}
\def\begdes{\begin{description}}
\def\enddes{\end{description}}
\def\begenu{\begin{enumerate}}
\def\begite{\begin{itemize}}
\def\endite{\end{itemize}}
\def\endenu{\end{enumerate}}

\def\lef[{\left[\begin{array}}
\def\rig]{\end{array}\right]}

\def\begcen{\begin{center}}
\def\endcen{\end{center}}
\def\begrem{\begin{remark}\rm}
\def\endrem{\end{remark}}

\newtheorem{theorem}{Theorem}[section]
\newtheorem{lemma}[theorem]{Lemma}

\newtheorem{proposition}[theorem]{Proposition}

\newtheorem{remark}[theorem]{Remark}
\newtheorem{assumption}[theorem]{Assumption}
\numberwithin{equation}{section}

\begin{document}
\title{\mbox{Modeling and Control of HVDC Transmission Systems}\\{From Theory to Practice and Back}}

\author[lss]{Daniele Zonetti\corref{cor1}}
\ead{daniele.zonetti@lss.supelec.fr}
\author[lss]{Romeo Ortega}
\ead{romeo.ortega@lss.supelec.fr}
\author[als]{Abdelkrim Benchaib}
\ead{abdelkrim.benchaib@alstom.com}

\cortext[cor1]{Corresponding author} \address[lss]{Laboratoire des Signaux et Syst\`{e}mes de Centrale Sup\'{e}lec, 3, rue Joliot Curie - 91192, Gif-sur-Yvette, France}
\address[als]{ Alstom Grid, 102, avenue de Paris - 91300 Massy, France}


\begin{abstract} 
The problem of modeling and control of multi--terminal high--voltage direct--current transmission systems is addressed in this paper, which contains five main contributions. First, to
propose a unified, physically motivated, modeling framework---based on port--Hamiltonian representations---of the various  network topologies used in this application. Second, to prove that the system
can be globally asymptotically stabilized with a decentralized PI control, that exploits its passivity properties. Close connections between the proposed PI and the popular Akagi's PQ instantaneous power method
are also established. Third, to reveal the transient performance limitations of the proposed controller that, interestingly, is shown to be intrinsic to PI passivity--based control. Fourth, motivated by the latter, an outer--loop that overcomes the aforementioned limitations is proposed. The performance limitation of the PI, and its drastic improvement using outer--loop controls, are verified via simulations on a three--terminals benchmark example. A final contribution is a novel formulation of the power flow equations for the centralized references calculation.
\end{abstract}

\begin{keyword}
multi--terminal HVDC transmission systems; passivity--based control; port--Hamiltonian systems; PI control;  nonminimum--phase systems; PQ and DC voltage control; performance limitations; power flow equations.
\end{keyword}

%

\maketitle

\section{Introduction}\label{intro}
In the last few years it has been observed an ever widespread utilization of renewable energy utilities, mainly based on wind and solar power \cite{jager,ANDERS}. Because of its intermittent nature the integration of
this generating units to the existing alternating--current (AC) distribution network poses a challenging problem \cite{carra,lund}. For this, and other reasons related to reduced losses and problems with reactive
power and voltage stability in AC systems, the option of high--voltage direct--current (HVDC) transmission systems is gaining wide popularity, see \cite{ANDERS,kirby,joha} for additional motivations and details.

The main components of an HVDC system are  AC to DC power converters, transmission lines and voltage bus capacitors. The power converters connect the AC sources---that are associated to renewable generating
units or to AC grids---to an HVDC grid through voltage bus capacitors. Two notable features distinguish HVDC systems from standard AC ones: the absence of a global signal (the synchronization frequency) and the central role
played by the power converters, the dynamics of which are highly {\em nonlinear}.

For its correct operation, HVDC systems---like all electrical power systems---must satisfy a large set of different regulation objectives that are, typically, associated to the multiple time--scale behavior of
the system. One way to deal with this issue, that prevails in practice, is the use of hierarchical architectures. These are nested control loops, at different time scales, each one providing references for an
inner controller \cite{kazm,iravani}. This paper focuses mainly on the ``innermost" control loop for HVDC transmission systems, that is, the control at the power converter level---in the
sequel referred as {\em inner--loop} control. The objective of the inner--loop control is to asymptotically drive the HVDC system towards a desired steady--state regime determined by the
user. Regulation should be achieved selecting a suitable switching policy for the converters. A major practical constraint is that the control should be {\em decentralized}. That is, the controller of each
power converter has only available for measurement its corresponding coordinates, with no exchange of information between them.

Starting from single AC/DC converter models many strategies have been proposed for the inner--loop control of the power converters used in HVDC systems \cite{akagi,lee,pinto}, with the dominating structure consisting of nested PI loops: an inner current control loop and an outer loop to regulate the capacitor voltage. The rationale used to justify this structure is the time--scale separation between currents and voltages. However, with the notable exception of \cite{torres}, the performance claims are not corroborated by rigorous stability proofs. Because of the absence of theoretical analysis, a time--consuming and expensive procedure to tune the gains of the PIs is then required to complete the design. This is typically done based on the linearization of the system that, because of the {highly} nonlinear behavior of the converters and the wide range of the operating regimes, often yields below--par performances.

The main objective of this paper is to contribute, if modestly, towards the development of a general, theoretically--founded procedure for the modeling, analysis and control of HVDC systems. With the intention to bridge the gap between theory and applications, one of the main concerns is to establish connections between existing engineering solutions, usually derived via  {\em ad--hoc} considerations, and the solutions stemming from theoretical analysis. In particular, it is shown that modifying the theoretically--based inner--loop controller to incorporate the standard considerations of outer--loop control considerably improves its transient performance.

The contributions of the paper are the following.
\begdes
\item[(C1)] To propose a unified, physically motivated, modeling framework of the various network topologies used in HVDC systems. This framework is based on port--Hamiltonian (pH) models of the system components \cite{ESCVANORT,VAN,ZON} combined with a suitable graph theoretic representation of their interconnection \cite{SHAetal}. The lines are modeled as simple series resistance--inductance ($RL$) circuits and the capacitors are assumed to be leaky elements, all components being linear. Although many different kinds of power converters are used in applications the dominant structure is the so--called voltage--source rectifier (VSR), which are the ones considered in the paper. The
network is described via a {\em meshed topology}, which allows for possible direct connection of the VSRs with the transmission lines.

\item[(C2)] In the spirit of \cite{HERetal,JAYetal,SANVER} it is proved that the incremental model of the VSR defines a {\em passive} map
with respect to some suitably designed output. A consequence of this fundamental property is that a decentralized PI passivity--based controller (PBC) {\em globally asymptotically stabilizes} (GAS) any assignable
equilibrium, with no restriction imposed on the (positive) gains of the PI--PBC. It is also shown that the proposed PI--PBC is closely related with Akagi's PQ instantaneous power method \cite{akagi} that was
derived (without a stability analysis) invoking power balance considerations and is standard in applications.

\item[(C3)] It is well--known that passive systems are minimum phase and have relative degree one \cite{BYRISIWIL,VAN}. Consequently, the attainable performance of a PI--PBC is limited by its associated zero dynamics. Another contribution of the paper is the proof that, in HVDC systems, the zero dynamics is ``extremely slow", stymying the achievement of fast transient responses. On the other hand, it is also shown that other inner--loop PI controllers reported in the literature may exhibit unstable behavior because the zero dynamics associated to the corresponding outputs are {\em non--minimum phase}.

\item[(C4)] Common engineering practice is followed to improve the transient performance, by adding an outer--loop that determines the PI--PBC reference signals---the widely diffused \textit{droop control} \cite{sandberg,sun}. After revisiting its standard formulation, a modification of the standard PI--PBC is proposed, showing that the intrinsic performance limitation are overcome, further preserving global asymptotic stability. The drastic improvement with this outer controller is finally verified via simulations on a three--terminals benchmark example.

\item[(C5)] A final contribution relates to the design of the last outer--loop controller in terms of a centralized reference calculator. Although there is no universal agreement to define the tasks of this control loop it usually relates to the regulation of the flow of active and reactive power to be injected into the network while keeping the voltage of the capacitors near a desired constant value. Most popular approaches, which usually invoke {\em ad--hoc} considerations, are reviewed and contextualized in the present framework.

\enddes

The remaining part of the paper is structured as follows. In Section \ref{modsection}, the mathematical model of the system is established (C1). Then, to determine the achievable behaviors, a study of the assignable equilibria is necessary. This analysis is done in Section \ref{sssection}. The main contribution (C2) is next developed in Section \ref{main}, with the design of the decentralized passivity--based PI controller. Slow transients exhibited in simulations motivate the subsequent performance analysis (C3), that is carried out in Section \ref{zdsection}. Sections \ref{sec6} and \ref{sec7} are then devoted to revisit standard outer--loop controllers (C4) and the problem of references calculation (C5). Conclusions and future work follow then in Section \ref{sec8}.\\

\noindent {\bf Notation} All vectors are column vectors. Given positive integers $n$, $m$, symbols  $\underline 0_n\in\rea^n$ denotes the vector of all zeros, $\mathbbm{1}_n\in\rea^n$ the
vector with all ones, $\mathbb{I}_n$ the $n \times n$ identity matrix, $\underline 0_{n\times m}$ the $n \times m$ column matrix of all zeros. $x:=\col(x_1,\dots,x_n)\in\rea^n$ denotes a
vector with entries $x_i \in \rea$, when clear from the context it is simply referred as $x:=\col(x_i)$. $\diag\{a_i\}$ is a diagonal matrix with entries $a_i \in \rea$ and $\text{bdiag}\{A_i\}$
denotes a block diagonal matrix with entries the matrices $A_i$. For a function $f:\rea^n \to \rea$, $\nabla f$ denotes the transpose of its gradient. The subindex $i$, preceded by a comma
when necessary, denotes elements corresponding to the $i$-th subsystem. 
%
\section{Energy--based Modeling}
\label{modsection}
%
In \cite{SHAetal} it was shown that electrical power systems can be represented by a directed graph\footnote{A directed graph is an ordered 3-tuple, $\mathcal{G}=\{\mathcal{\mathcal V,\mathcal E},\Pi\}$,
consisting of a finite set of nodes $\mathcal{V}$, a finite set of directed edges $\mathcal{E}$ and a mapping $\Pi$ from $\mathcal{E}$ to the set of ordered pairs of $\mathcal{V}$, where no self-loops are
allowed.} where the relevant electrical components correspond to edges and the buses correspond to nodes. Moreover, to underscore the physical structure of the components, they are modeled as pH systems. In
this section the same procedure is applied to describe the dynamics of HVDC transmission systems.
%
\subsection{Assumptions}
As indicated in the Introduction, the relevant components for an HVDC transmission system are: VSRs, $RL$ transmission lines and voltage bus capacitors.  Throughout the paper the
following assumptions---which are widely accepted in practice---are made.
\begdes
\item[(A1)] Balanced operation of the three phase line voltages.
\item[(A2)] Synchronized operation  of the VSRs\footnote{Synchronized operation  of the VSRs is usually achieved via robust phase--locked--loop detection ofthe latching frequencies \cite{iravani}.}.
\item[(A3)] Ideal four quadrant operation of the VSRs.
\enddes

Assumptions A1 and A2 considerably simplify the modeling and control problems, as they allow the description of the three--phase dynamics of the VSRs in suitably oriented $dq0$ reference frames, where the value
of the $0$--component is always zero, thus reducing the three AC quantities to two DC quantities. Consequently, it is possible to express the regulation objective as a standard {\em equilibrium stabilization}
problem of the nonlinear dynamical system describing the behavior of the HVDC system. Assumption A3 directly follows by assuming an HVDC transmission system based on VSRs instead of current source rectifiers,
which is an alternative converter topology used in HVDC systems. As a matter of fact, since the VSRs do not depend on line--commutations, all the four quadrants of the operating plane are possible, hence
Assumption A3 is automatically satisfied for the system under consideration  \cite{abbas}.

%
\subsection{Network topologies: A graph description}
It is possible to distinguish two kinds of topologies used in  HVDC transmission systems: \textit{radial} and \textit{meshed} topology \cite{agelidis,bucher,gomis}, which are illustrated in Fig. \ref{topo}. The
radial topology is widely used for systems in which a certain number of off--shore stations feeds on--shore stations with no connection between them. This is the case for example of on--shore stations situated
on opposite seacoasts while the off--shore stations are placed in their middle \cite{bucher,kirby}. However, in a more general setting one has to consider the situation in which the stations are directly connected with lines, that corresponds to a meshed topology. In the interest of brevity, a systematic
way to build global pH models is presented only for the {meshed} topology. For a {radial topology}, analogous results can be obtained, for which the interested reader is referred to \cite{ECCZonetti}.

\begin{figure}[ht]
 \centering
 \includegraphics[width=1\linewidth]{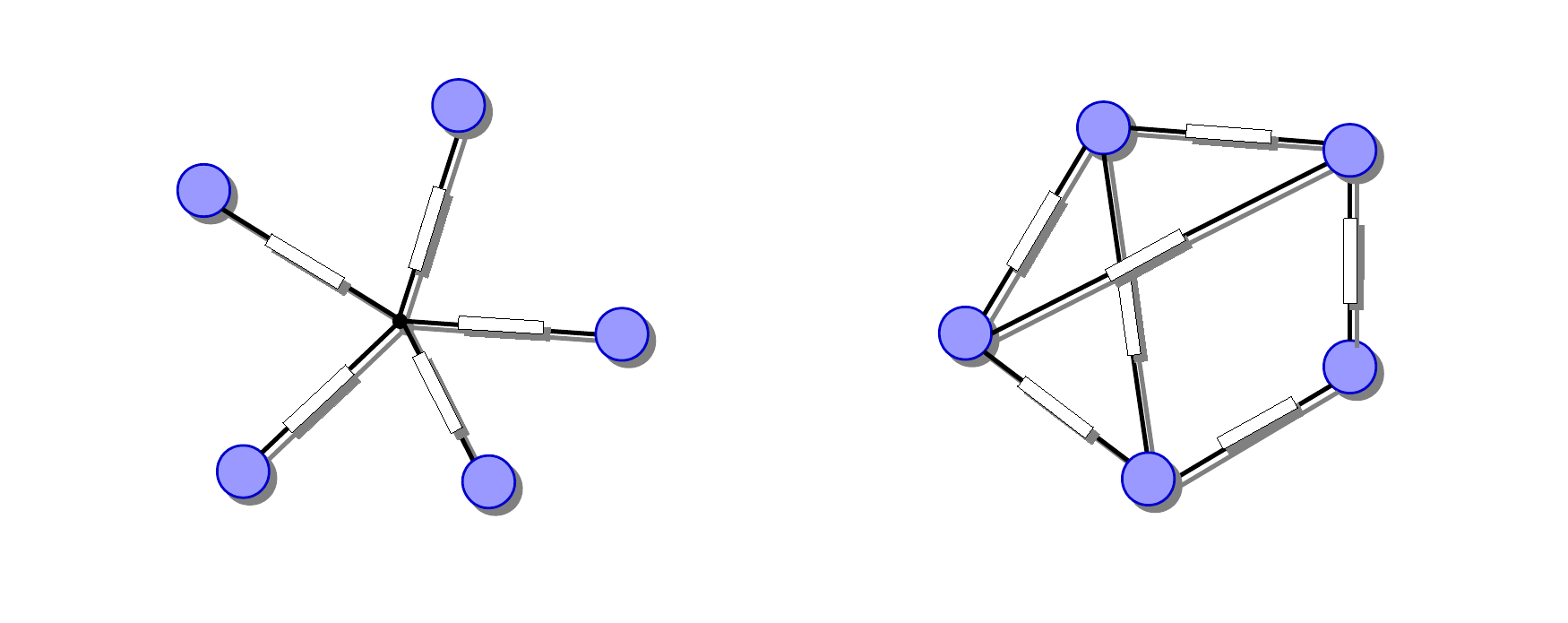}
 \caption{Nodal representation of HVDC transmission systems with radial and meshed topologies.}
 \label{topo}
\end{figure}

In order to give a formal representation of a topology the following definitions are adopted. A bus is called a \textit{VSR-bus} if a VSR is connected to it and a bus is called a
\textit{capacitor-bus} when only a capacitor is connected to it. Furthermore, a bus is called a \textit{reference-bus} when all the voltages of the buses in the network are measured with
respect to it. As the {reference-bus} is assumed to be at ground potential, it is also denoted as \textit{ground}. A general topology is then described by the incidence matrix $
\mathcal M$ associated
to the graph, where the nodes represent the {ground}, the {VSR} and the {capacitor-buses}; the edges represent the VSRs, the lines and the single capacitors that are interconnected to the
{ground} or to the voltage buses.

In a {meshed} topology, each VSR is connected to the {ground} and to a {VSR--bus}, while the lines directly connect {VSR--buses}, according to a determined meshed structure. The
number $n$ of VSRs is the same as voltage buses, {ground} excluded, and is lower or equal to the number $\ell$ of lines. Formally, this can be represented by a graph $\mathcal
G:=\{\mathcal{V},\mathcal{E},\Pi\}$ constituted by: $n+1$ ordered nodes, where $n$ nodes are associated to the {VSR--buses} and one node to the {ground}; $n+\ell$ ordered edges, where $n$
edges are associated to the VSRs and $\ell$ edges to the lines. The incidence matrix then --- following the mentioned order --- takes the form
\begin{equation}
\label{mdelta}
\mathcal M=\begin{bmatrix}
\mathbb{I}_n&M \\
-\mathbbm{1}^\top_{n}&\underline{0}^\top_\ell\\
\end{bmatrix}\in\mathbb{R}^{(n+1)\times (n+\ell)},
\end{equation}
where $M$ is the incidence matrix of the subgraph obtained by eliminating the VSR edges and the ground node.
\begin{remark}\em
\label{parallel}
In a {meshed} topology the only relevant components are the VSRs and the $RL$ transmission lines. As a matter of fact, because a VSR is associated to each node, the voltage bus capacitors
can be represented by an equivalent VSR output capacitor, that results to be the parallel interconnection of all capacitors attached to the node.
\end{remark}
%
\subsection{Port--Hamiltonian models of the elements}
As explained above, the edges of the graph $\mathcal{G}$ contain the electrical components of the HVDC system, namely $n$ VSRs and $\ell$ $RL$ transmission lines, while the nodes are the buses. In this section,
a pH representation of each of these elements is derived, which are then interconnected---through power preserving interconnections---via the graph. Besides its physically appealing nature, the choice of a pH model is
motivated by the fact that---similarly to \cite{HERetal}---this structure is instrumental to derive the passivity property exploited in the controller design. To enhance readability the models of the VSRs and
the transmission lines are presented separately.
%
\subsubsection{Voltage source rectifiers}
In \cite{ESCVANORT,HERetal,ECCZonetti} the well--known average model of a single VSR shown in Fig. \ref{rectifier}, expressed in $dq$--coordinates and written in (perturbed) pH form is given. Similarly, a {\em set} of $n$ VSRs can also be represented in pH form as
\begin{equation}
\label{rectifierPH}
\begin{aligned}
\dot x_{R}&=[\mathcal{J}_{R}(u)-\mathcal{R}_{R}]\nabla\mathcal H_{R}+E_1V-E_3i_{R}\\
v_{R}&=E_3^\top \nabla\mathcal H_{R},
\end{aligned}
\end{equation}
with the following definitions.
\begin{itemize}
\item[-] State space variables the collection of inductors fluxes $(\phi_{d,i}, \phi_{q,i})$ and capacitor charges $q_{c,i}$ of every VSR, that is, $x_{R}:=\col(\col(\phi_{d,i}), \col(\phi_{q,i}), \col(q_{c,i})) \in\mathbb{R}^{3n}$.

\item[-] Energy function 
\begin{equation*}
\mathcal{H}_{R}(x_{R}):=\frac{1}{2} x_{R}^\top Q_{R} x_{R},\quad Q_R:=\mathrm{bdiag}\{L_R^{-1},L_R^{-1},C_R^{-1}\},
\end{equation*}
with
$$
L_{R}:=\mathrm{diag}\{L_{r,i}\} ,\quad  C_{R}:=\mathrm{diag}\{C_{r,i}\},
$$

where $L_{r,i},C_{r,i}$ are the inductance and capacitance of each VSR, respectively\footnote{Unless indicated otherwise all physical parameters of the system are positive constants.}.

\item[-] Duty cycles $u:=\col(u_{Rd},u_{Rq}) \in\mathbb{R}^{2n}$, where $u_{Rd}:=\col(u_{d,i})$ and $u_{Rq}:=\col(u_{q,i})$.

\item[-] External voltage sources $V:=\col (v_{d,i})\in\mathbb{R}^n$, where $v_{d,i}$ is the $d$ component of the AC sources. These voltages are assumed constant and positive.

\item[-] Port variables the out--going currents $i_{R}:=\col(i_{dc,i})\in\mathbb{R}^n$ and the voltages $v_{R}:=\col(v_{dc,i})\in\mathbb{R}^n$.

\item[-] Interconnection matrix
\begequ
\label{Jdecomp}
\mathcal{J}_{R}(u):=\sum_{i=1}^n(\mathcal{J}_{R0,i}L_{r,i}\omega_i+\mathcal{J}_{Rd,i}u_{d,i}+\mathcal{J}_{Rq,i}u_{q,i})
%
 \endequ
 where $\omega_i $ are the AC sides frequencies and
\begin{align*}
&\mathcal{J}_{R0,i}:=\begin{cases}-1\;\mathrm{in}\; _{(i,n+i)}\\
 1\;\mathrm{in}\; _{(n+i,i)}\\
0\;\mathrm{elsewhere}
\end{cases}\\
&  \mathcal{J}_{Rd,i} :=\begin{cases}1\;\mathrm{in}\; _{(i,2n+i)}\\
 -1\;\mathrm{in}\; _{(2n+i,i)}\\
0\;\mathrm{elsewhere}
\end{cases} 
\mathcal{J}_{Rq,i} :=\begin{cases}-1\;\mathrm{in}\; _{(n+i,2n+i)}\\
 1\;\mathrm{in}\; _{(2n+i,n+i)}\\
0\;\mathrm{elsewhere}
\end{cases}
\end{align*}
\item[-] Dissipation matrix $\mathcal{R}_R:=\mathrm{bdiag}\{R_R ,R_R ,G_R \}$, where $R_{R}:=\diag\{R_{r,i}\}$ and $G_{R}:=\diag\{G_{r,i}\}$ , with $R_{r,i},G_{r,i}$ the resistance and conductance of each VSR.

\item[-] Port matrices $E_1:\begin{bmatrix}
\mathbb{I}_n&0&0
\end{bmatrix}^\top$, $E_3:=\begin{bmatrix}
0&0&\mathbb{I}_n
\end{bmatrix}^\top\in\mathbb{R}^{3n\times n}$.
\end{itemize}
\begin{figure}[ht]
 \centering
\includegraphics[width=1\linewidth]{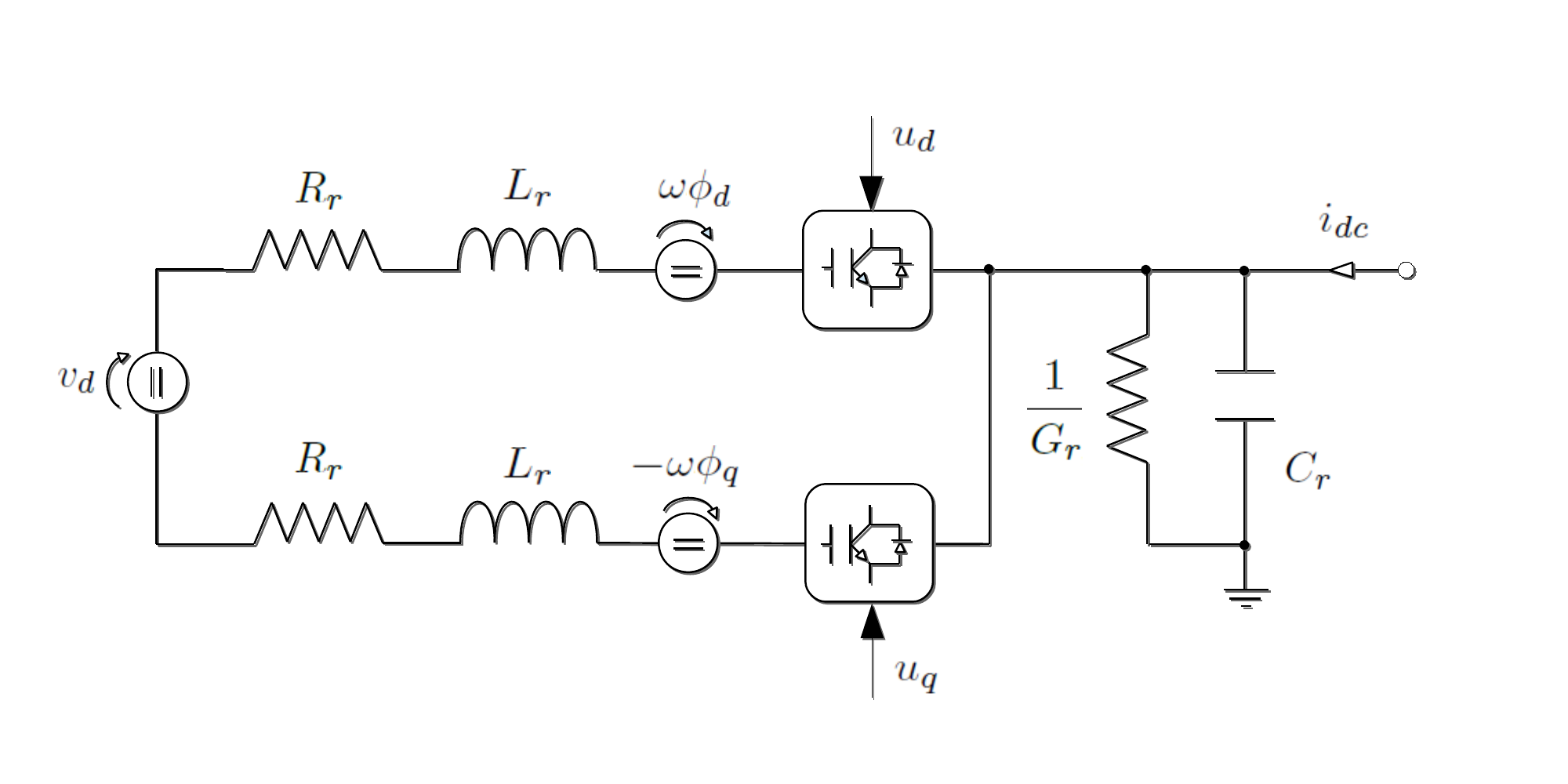}
 \caption{Schematic diagram of the equivalent circuit of a VSR in $dq$ frame.}
 \label{rectifier}
\end{figure}

\begrem
Note that, in view of the skew--symmetry of $\mathcal{J}_{R}(u)$, the VSRs satisfy the power balance equation
\begequ
\lab{powbal}
\underbrace{\dot {\mathcal H}_R}_{\mbox{stored\:power}}=-\underbrace{x_R^\top Q_R\mathcal{R}_RQ_Rx_R}_{\mbox{dissipated\;power}}+\underbrace{x_R^\top Q_RE_1V-x_R^\top Q_RE_3i_R}_{\mbox{supplied\;power}}
\endequ
\endrem
%
\subsubsection{Transmission lines}
A set of $\ell$ $RL$ transmission lines can be represented by the pH system
\begin{equation}
\begin{aligned}
\dot x_{L} &=-\mathcal{R}_{L}\nabla\mathcal H_{L}+v_{L}\\
i_{L}&=-\nabla\mathcal H_{L},
\end{aligned}
\label{PHline}
\end{equation}
with  the following definitions.
\begin{itemize}
\item[-] State space variables the collection of inductor fluxes $x_{L}:=\mathrm{col}(\phi_{\ell,i})\in\mathbb{R}^\ell$ of every line.
\item[-] Energy function
\begin{equation*}
\mathcal H_{L}(x_{L}):=\frac{1}{2} x_L^\top Q_{L} x_{L},\qquad Q_{L}:=\mathrm{diag}\{\frac{1}{L_{\ell,i}}\},
\end{equation*}
where $L_{\ell,i}$ is the inductance of the line.

\item[-] Port variables the voltages at the terminals $v_{L}:=\mathrm col(v_{L,i})\in\mathbb{R}^\ell$ and the inductors currents $i_{L}:=\mathrm{col}(i_{\ell,i})\in\mathbb{R}^\ell$.

\item[-] Dissipation $\mathcal{R}_{L}=\mathrm{diag}\{R_{\ell,i}\}$, with $R_{\ell,i}$ the resistance of the line.
\end{itemize}
%
\subsection{Overall interconnected system}
\lab{subsec2.4}
The interconnection laws can be obtained following the approach used in \cite{avdspartial}, where Kirchhoff's current and voltage laws (KCL and KVL, respectively) are expressed in relation to the incidence matrix. For a \textit{meshed} topology then it follows
\begin{equation}\label{interlaws}
\begin{aligned}
\textsf{[KCL]}\qquad&\mathcal M\mathcal{I}_e=\underline{0}_{n+1},\\
\textsf{[KVL]}\qquad&\mathcal M^\top\mathcal{V}=\mathcal{V}_e,
\end{aligned}
\end{equation}
where $\mathcal{I}_e:=\col (i_R,i_L)$, $\mathcal{V}_e:=\col (v_R,v_L)$ and $\mathcal{V}:=\col (v_1,\dots,v_{n})$, $v_0$ are the edge currents, the edge voltages, the nodes potentials and
the {ground} potential, respectively. The {ground} potential $v_0=0$ by definition. From \eqref{interlaws} and \eqref{mdelta} then follows
\begin{equation}
\lab{iril}
\begin{aligned}
i_R+M i_L&=\underline{0}_n,\quad -\mathbbm{1}_n^\top i_R=0,\\
v&=v_{R},\quad  M ^\top v=v_L.
\end{aligned}
\end{equation}
Recalling the expression for $i_L$ from \eqref{rectifierPH} and $v_R$ from \eqref{PHline} it is easy to get
\begin{equation}\label{intlaws2}
i_R=M \nabla\mathcal{H}_L,\qquad v_L=M ^\top E_3^\top\nabla\mathcal{H}_R.
\end{equation}
To obtain the overall pH representation it is then sufficient to combine \eqref{rectifierPH}, \eqref{PHline} and \eqref{intlaws2}, thus leading to
\begin{equation}\label{overall}
\dot{x}=[\mathcal{J}(u)-\mathcal{R}]\nabla\mathcal{H}+EV,
\end{equation}
with  the following definitions.
\begin{itemize}
\item[-] State space variables $x:=\col(x_R,x_L) \in\mathbb{R}^{3n+\ell}$.

\item[-] Energy function $\mathcal H(x):=\mathcal{H}_R(x)+\mathcal{H}_L(x)$.

\item[-] Duty cycles (controls) $u:=\col(u_{Rd},u_{Rq}) \in\mathbb{R}^{2n}$.

\item[-] Interconnection matrix
\begin{equation}\label{Joverall}
\mathcal{J}(u):=\begin{bmatrix}J_{R}(u) & -E_3M \\ M ^\top E_3^\top&\underline{0}_{\ell\times
\ell}\end{bmatrix}, 
\end{equation}

\item[-] Dissipation matrix
\begequ
\lab{dismat}
\mathcal{R}:=\mathrm{bdiag}\{\mathcal{R}_{R},\mathcal{R}_{L}\}>0.
\endequ
\item[-] Input matrix $E:= \begin{bmatrix}
E_1^\top \underline{0}_{\ell\times n}^\top
\end{bmatrix}^\top.$
\end{itemize}

\begrem
To simplify the notation in the pH representation it is selected a state representation of the system using energy variables, that is, inductor fluxes and capacitor charges, instead of the more commonly used
co--energy variables, {\em i.e.}, inductor currents and capacitor voltages. See \eqref{sysiv} and \cite{perez} for the derivation of the pH model in the latter coordinates. The coordinates are indeed related by
\begin{equation}\label{coenergy}
i_d=\frac{\phi_d}{L},\quad i_q=\frac{\phi_q}{L},\quad v_C=\frac{q_C}{C},\quad i_L=\frac{\phi_L}{L_\ell}.
\end{equation}
\endrem

\begrem
For ease of presentation it is assumed that the state of the system lives in $\mathbb{R}^{3n+\ell}$. Due to physical and technological constraints it is actually only defined in a subset of
$\mathbb{R}^{3n+\ell}$. In particular, the voltage of the DC link $v_C$ is strictly bounded away from zero.
\endrem
%
\section{Assignable Equilibria}
\label{sssection}
%
A first step towards the development of a control strategy for the system \eqref{overall} is the definition of its achievable, steady--state behavior, which is determined by the assignable
equilibria. That is, the (constant) vectors $x^\star \in\mathbb{R}^{3n+\ell}$ such that
$$
[\mathcal{J}(u^\star )-\mathcal{R}]\nabla\mathcal H(x^\star )+EV=\underline{0}_{3n+\ell}
$$
for some (constant) vector $u^\star \in \rea^{2n}$. To identify this set the following lemmata are established.

\begin{lemma}\em
\lab{lem1}
The equilibria of the transmission line coordinates are given by
\begequ
\lab{equ1}
x^\star _L=(\mathcal{R}_LQ_L)^{-1}M^\top E_3^\top Q_Rx_R^\star.
\endequ
\end{lemma}
\begin{proof}
Setting to zero the left--hand side of \eqref{PHline}, calculated at $x_L^\star $, gives
$$
\underline{0}_\ell=-\mathcal{R}_LQ_Lx^\star _L+v_L^\star \quad\Rightarrow\quad x^\star _L=(\mathcal{R}_LQ_L)^{-1}v_L^\star .
$$
Moreover, from \eqref{intlaws2} it follows that $v_L^\star =M^\top E_3^\top Q_Rx_R^\star $, that replaced in the equation above completes the proof.
\end{proof}

\begin{lemma}\em
\lab{lem2}
The equilibria of the VSRs coordinates are the solution of the $n$ quadratic equations, $i=1\dots n$
\begequ
\lab{equ2}\resizebox{1\hsize}{!}{$
-\frac{R_i}{L_{r,i}^2} \left[(\phi_{d,i}^\star)^2+(\phi_{q,i}^\star)^2\right]-\frac{G_i}{C_{r,i}^2}(q_{C,i}^\star)^2+\frac{v_{d,i}}{L_{r,i}}\phi_{d,i}^\star-\frac{1}{C_{r,i}}q^\star_{C,i}i_{dc,i}^\star=0,$}
\endequ
with $
\mathrm{col}(i^\star_{dc,i})=M\mathcal{R}_L^{-1}M^\top \mathrm{col}(q_{C,i}^\star).
$
\end{lemma}
\begin{proof}
In \cite{sanchez} it is shown that the set of admissible equilibria of a VSR is obtained by setting equal to zero the power balance of the VSR, that for $n$ VSRs is equivalent to \eqref{equ2}. To complete the proof, it is now sufficient to recall definitions $$\mathrm{col}(i^\star_{dc,i})=i_R^\star,\qquad E_3^\top Q_Rx_R^\star=\mathrm{col}(q_{C,i}^\star),$$ together with \eqref{intlaws2}, \eqref{equ1}.
\end{proof}
The main result of the section is now presented, the proof of which follows immediately from the lemmata above.
\begin{proposition}\em
\label{prop}
The set of assignable equilibria of the system \eqref{overall} is given by
\begequ
\mathcal{E}:=\{ x^\star  \in\mathbb{R}^{3n+\ell}\;\vert \; \eqref{equ1}\;\mbox{and}\; \eqref{equ2}\; \mbox{hold}\}.
\endequ

\end{proposition}

From the derivations above it is clear that the equilibria of the network are univocally determined by the equilibria of the VSRs. Moreover, the latter should satisfy the quadratic equations \eqref{equ2},
which are the well--known \textit{power flow steady--state equations} (PFSSE) of the individual VSR subsystems. A question of  interest is how to select from this set the equilibrium points that correspond to some {\em desired behavior}. In the latter definition there are many practical considerations to be taken into account, the discussion of which is postponed to Sections \ref{sec6} and \ref{sec7}.

\begin{remark}\em
\label{remU}
It is well--known that for affine systems of the form $\dot x = f(x)+g(x)u$ the assignable equilibrium set is given by
$$
\{ x^\star  \in \rea^n\;\vert\; g^\perp(x^\star )f(x^\star )=0\}
$$
where $g^\perp(x)$ is a full--rank left annihilator of $g(x)$. Moreover, given $x^\star $, the corresponding equilibrium control $u^\star $ is {\em univocally} determined by
$$
u^\star =-\left[ (g^\top g)^{-1}g^\top f\right] (x^\star ).
$$
Since \eqref{overall} is clearly of this form this relations hold true for the HVDC transmission system under study. See  \cite{sanchez} for additional details on this issue.
\end{remark}

\begin{remark}\em
Differently from the single VSR case, the set of assignable equilibria does not coincide, but is strictly contained, in the set where the power of the system is balanced, that is
$$
\mathcal{E}\subset\mathcal{P},\qquad \mathcal{P}:=\{x^\star\in\mathbb{R}^{3n+\ell}\;\vert\;\dot{\mathcal{H}}_R=0\}.
$$
This fact is clearly explained in \cite{sanchez}, where it is proved that a necessary condition for $\mathcal{E}\equiv\mathcal{P}$, is the system to be of co-dimension one.
\end{remark}

%
\section{Main Result: Inner Loop Control}
\label{main}
%
As indicated in the Introduction, this paper mainly focuses on the inner--loop control of HVDC transmission systems, that is, the control at the VSR level. For, in this
section it is presented a decentralized, globally asymptotically stabilizing, PI--PBC for the HVDC transmission system \eqref{overall}. The construction of the controller is inspired by previous works of the authors on PI--PBC, reported
in \cite{HERetal} and \cite{JAYetal}, which exploit the property of passivity of the \textit{incremental model}. The interested reader is referred to these references for additional details.

As indicated above, it is assumed that a desired operating point $x^\star \in \mathcal{E}$ has already been selected---further discussions on its choice are deferred to Sections \ref{sec6} and \ref{sec7}. To place the proposed PI--PBC in context, in the last part of this section the most commonly used inner--loop controls for HVDC transmission systems are  briefly reviewed and a connection is established with the widely popular Akagi's PQ method.
%
\subsection{Passivity of the incremental model}
\label{PBcontrol}
Along the lines of Proposition 1 in \cite{HERetal}, it is possible to establish passivity of the incremental model of the overall HVDC transmission system \eqref{overall} with respect to a suitable defined
output. As is well--known, global regulation of a passive output can be achieved with a simple PI controller. Regulation of the state to the desired equilibrium then follows provided a suitable detectability
assumption is satisfied \cite{VAN}.\smallbreak

\begin{proposition}\em
Consider the HVDC transmission system \eqref{overall}. Let $x^\star  \in\mathcal{E}$ be the desired equilibrium with corresponding (univocally defined) control $u^\star \in\mathbb{R}^{2n}$. Define the error signals
\begequ
\lab{tilxu}
\tilde x = x-x^\star ,\quad \tilde u = u - u^\star
\endequ
and the output signal
\begequ\label{y}
 y:=\begin{bmatrix}
\mathrm{col}(y_{d,i})\\
\mathrm{col}(y_{q,i})
\end{bmatrix}\in\mathbb{R}^{2n},
\endequ
with
\begin{equation*}
\begin{aligned}
y_{d,i}:=x_R^{*\top} Q_{R}\mathcal{J}_{Rd,i}Q_{R}x_{R},\qquad y_{q,i}:=x_R^{*\top} Q_{R}\mathcal{J}_{Rq,i}Q_{R}x_{R}.
\end{aligned}
\end{equation*}
The mapping $\tilde u \to y$ is {\em passive}. More precisely, the system verifies the dissipation inequality
\begin{equation}
\lab{dothd}
\dot{\mathcal{H}}_{d}\leq y^\top \tilde u,
\end{equation}
with storage function
$
\mathcal{H}_{d}(\tilde x)=\frac{1}{2}\tilde x^\top Q \tilde x.$
\end{proposition}
\begin{proof}
The proof mimics the proof of Proposition 1 in \cite{HERetal}. First of all, it is possible to write
$$
\mathcal{J}(u)Qx=\mathcal{J}_0Qx+g(x) u,
$$
where the following definitions
$$ \resizebox{1\hsize}{!}{$
\mathcal{J}_0:=
\begin{bmatrix}\sum_{i=1}^n(\mathcal{J}_{R0,i}L_{r,i}\omega_i) & -E_3M \\ M ^\top E_3^\top&\underline{0}_{\ell\times
\ell}\end{bmatrix}, 
\quad
g(x):=\begin{bmatrix}
g_{Rd}(x_R) &g_{Rq}(x_R)\\
\underline{0}_{\ell\times n}&\underline{0}_{\ell\times n}
\end{bmatrix}, $}
$$
with \begin{align*}
g_{Rd}(x_R):&=
\begin{bmatrix}
\mathcal{J}_{Rd,1}Q_Rx_R&\dots&\mathcal{J}_{Rd,n}Q_Rx_R
\end{bmatrix},\\
g_{Rq}(x_R):&=\begin{bmatrix}
\mathcal{J}_{Rq,1}Q_Rx_R&\dots&\mathcal{J}_{Rq,n}Q_Rx_R
\end{bmatrix}
\end{align*}
are adopted. Hence, it is possible to write \eqref{overall} in the alternative form
\begin{equation}\label{incremental}
\begin{aligned}
\dot{x}&=(\mathcal{J}_{0}-\mathcal{R})Q x+EV+g(x)  u\\
&=(\mathcal{J}_{0}-\mathcal{R})Q(\tilde x+x^\star )+EV+g(x)  (\tilde u+u^\star )\\
&=(\mathcal{J}_0-\mathcal{R})Q\tilde x+g(x) \tilde u+g(\tilde x)  u^\star
\end{aligned}
\end{equation}
where \eqref{tilxu} has been used to get the second equation and the fact that the assignable equilibria $x^\star $ and corresponding (constant) control $u^\star $ satisfy
$$
(\mathcal{J}_{0}-\mathcal{R})Q x^\star +EV+g(x^\star )  u^\star =0,
$$
is used to obtain the third equation.

The derivative of $\mathcal{H}_d$ along the trajectories of the incremental model \eqref{incremental} yields
$$
\dot{\mathcal{H}}_d   =   -\tilde x^\top Q\mathcal RQ\tilde x +\tilde x^\top Qg(x)  \tilde u = -\tilde x^\top Q\mathcal RQ\tilde x + y ^\top\tilde u
$$
where the skew--symmetry of $\mathcal{J}_{0}$, $\mathcal{J}_{Rd,i}$ and $\mathcal{J}_{Rq,i}$ is used in the first equation, and the fact that the output signal   can be rewritten as
$$
y=g^\top(x^\star )Qx=g^\top(x^\star )Q \tilde x
$$
is used to obtain the second identity. The proof is completed recalling that the dissipation matrix verifies $\mathcal R > 0$  to obtain the bound  \eqref{dothd}.
\end{proof}
%
\subsection{PI passivity--based control}
\label{PBCcontrol}
The first main result of the paper is then presented. \smallbreak
\begin{proposition}\em
\lab{pro2}
Consider the HVDC transmission system \eqref{overall}, with a desired steady--state $x^\star  \in\mathcal{E}$, in closed--loop with the decentralized PI control
\begin{equation}\label{control}
\begin{aligned}
u & = -K_P y - K_I\zeta,\quad \dot \zeta = y,
\end{aligned}
\end{equation}
with $y$ given in \eqref{y} and gain matrices
\begin{equation}
\begin{aligned}\label{gains}
K_P=\begin{bmatrix}
K_{Pd}&0\\
0&K_{Pq}
\end{bmatrix}\in\mathbb{R}^{2n\times 2n},\quad 
K_I=\begin{bmatrix}
K_{Id}&0\\
0&K_{Iq}
\end{bmatrix}\in\mathbb{R}^{2n\times 2n},
\end{aligned}
\end{equation}
where $K_{Pd} =\mathrm{diag}\{k_{Pd,i}\}, K_{Pq}=\mathrm{diag}\{k_{Pq,i}\}$, $K_{Id} =\mathrm{diag}\{k_{Id,i}\}$, $K_{Iq}=\mathrm{diag}\{k_{Iq,i}\}$. The equilibrium point $(x^\star ,K_I^{-1}u^\star )$ is globally asymptotically stable (GAS).
\end{proposition}
\begin{proof}
Define the Lyapunov function candidate
\begin{equation}\label{weq}
W(\tilde x,\tilde \zeta):=\mathcal{H}_d(\tilde x)+\frac{1}{2}\tilde\zeta^\top K_I\tilde\zeta,
\end{equation}
where $
\tilde \zeta:=\zeta - K_I^{-1}u^\star .
$
The derivative of $W(x,\zeta)$ along the trajectories of the closed--loop system \eqref{overall}-\eqref{control} is given by
\begin{equation*}
\begin{aligned}
\dot{{W}}&=-\tilde x^\top Q\mathcal{R}Q\tilde x+{y}^\top\tilde{u}+\tilde\zeta^\top K_Iy\\
&=-\tilde x^\top Q\mathcal{R}Q\tilde x+{y}^\top\tilde{u}-(\tilde u^\top +y^\top K_P)y \\
& =  -\tilde x^\top Q\mathcal{R}Q\tilde x- {y}^\top K_P {y\leq 0},
\end{aligned}
\end{equation*}
which proves global stability. Asymptotic stability follows, as done in \cite{HERetal}, using LaSalle's arguments. Indeed, from the inequality above and the definition of $\mathcal{R}$ in \eqref{dismat} it is clear
that all components of the error state vector $\tilde x$ tend asymptotically to zero.
\end{proof}

\begrem
The proposed PI--PBC is decentralized in the sense that, for its implementation, each VSR control requires only the measurement of its corresponding inductor currents and capacitor voltage. Guaranteeing this property motivates the choice of block diagonal gain matrices \eqref{gains}.
\endrem

\begrem
The PI--PBC requires the selection of the desired values for the inductor currents and capacitor voltages that, clearly, cannot all be chosen arbitrarily. Instead, they have to be selected from the set of
assignable equilibrium points $\mathcal{E}$, that is determined by the {PFSSE}.  This set has a rather simple structure: the quadratic equation \eqref{equ2} defines the VSRs variables from which it is possible to {\em
univocally} determine the transmission lines coordinates via \eqref{equ1}.
\endrem
%
\subsection{Other inner--loop controllers reported in the literature}
\label{liter}
In this section some of the inner--loop controllers for VSRs reported in the literature are reviewed. The vast majority of the papers reported on this topic---and, in general, of control of power converters
\cite{kazm,pinto}---uses the description of the dynamics in co--energy variables. To facilitate the reference to these works, the following model---that is immediately
obtained from \eqref{rectifierPH} and \eqref{coenergy}---is provided\footnote{For ease of presentation the discussion here is restricted to a {\em single} VSR. The extension to multiple VSRs being straightforward.}:
\begin{equation}
\label{sysiv}
\begin{aligned}
L\dot{i}_d&=-Ri_d+L\omega i_q-v_Cu_d+v_d\\
L\dot{i}_q&=-L\omega i_d-Ri_q-v_Cu_q\\
C\dot{v}_C&=i_du_d+i_qu_q-Gv_C-i_{dc}.
\end{aligned}
\end{equation}
The total energy of the VSR is
$$
\mathcal H(i_d,i_q,v_C):=\frac{1}{2}\left( L{i}^2_d+L i^2_q+ Cv^2_C \right),
$$
and the power balance is
\begequ
\lab{powbalvsr}
\dot{\mathcal H} = -R(i^2_d+i^2_q)-Gv^2_C+P-P_{dc},
\endequ
where the active and DC powers are defined as
\begin{equation}\label{power}
P=v_di_d, \qquad P_{dc}=v_Ci_{dc}.
\end{equation}
It is also common to define the reactive power as $Q=v_di_q$.

A caveat regarding the subsequent analysis is, however, necessary. When the VSRs are connected to the transmission lines the currents $i_{dc}$ are linked to the currents on the line via \eqref{iril}, which are clearly nonconstant. However, to simplify the analysis, it is assumed that they are {\em constant}. This can be justified by exploiting the fact that their rate of change is slow (with respect to the VSR dynamics).
Under this assumption the assignable equilibrium set of \eqref{sysiv} is given as
\begin{equation}
\label{equset}
{\mathcal{E}}=\{x \in \rea^3\;|\;R(i_d^2+i_q^2) - v_d i_d+G v_C^2 + i_{dc} v_C=0\}.
\end{equation}

Since $v_d$ and $i_{dc}$ are constant, it is then clear that the regulation of $P$, $Q$ and $P_{dc}$ are equivalent to the regulation of $i_d$, $i_q$ and $v_C$, respectively. In practice, because of the small
losses of the VSR, the value of $P$ slightly differs from $P_{dc}$, and consequently there is no interest in regulating the pair $i_d$ and $v_C$ at the same time.

In the literature, it is common to distinguish two modes of operation for a VSR:
\begite
\item[-] \textit{PQ control mode}, when the VSR is required to control the active and reactive power. This is achieved regulating to zero the output
\begin{equation}\label{DCC}
y_I=\lef[{c} i_d-i_d^{\mathrm{ref}}\\i_q-i_q^{\mathrm{ref}}\rig],
\end{equation}
where the superscript $(\cdot)^{\mathrm{ref}}$ is used to denote reference values---that {\em do not} necessarily  belong to the assignable equilibrium set. These kind of schemes are also called \textit{direct
current control} \cite{shuai}.

\item[-] \textit{DC voltage control mode}, when the VSR is required to control reactive power and DC voltage. In this case, the regulated output is
\begin{equation}\label{DOVC}
y_V=\lef[{c} v_C-v_C^{\mathrm{ref}}\\i_q-i_q^{\mathrm{ref}}\rig].
\end{equation}
These kind of schemes are also called \textit{direct
output voltage control} \cite{shuai}.
\endite

To regulate the outputs \eqref{DCC} and \eqref{DOVC} different controllers have been proposed in the literature, ranging from simple PI control \cite{lee,pinto} to feedback linearization
\cite{chen1,chen2,bencha}. Some of these papers include (invariably local) stability analysis. In Section \ref{zdsection} it is proved that $y_I$  and $y_V$, used for the PIs or with respect to
which feedback linearization is performed, have {\em unstable zero dynamics}. Consequently, applying high gains in the PIs will induce instability and the internal behavior of the feedback linearizing schemes
will be unstable.\footnote{This well--known phenomenon of nonlinear systems \cite{isidori} is akin to cancellation of unstable zeros of the plant with the unstable poles of the controller in linear
systems.} Simulations in Subsection \ref{simuinner} show that instability indeed arises for these schemes.

For the sake of comparison the passive output \eqref{y} in co--energy variables, for a single VSR, is provided:
\begin{equation}\label{yy}
y=\lef[{c} v^\star _{C}i_{d}-i^\star _{d}v_{C} \\ v^\star _{C}i_{q}-i^\star _{q}v_{C} \rig],
\end{equation}
where $(i^\star _{d},i^\star _{q},v^\star _{C}) \in {\mathcal{E}}$, that is, they belong to the assignable equilibrium set.

\begrem\label{universal}
The PI--PBC is \textit{universal}, in the sense that it can operate either in \textit{PQ} or \textit{DC voltage control mode}, depending on which equilibria are assigned as desired references, and which one is consequently determined via the PFSSE. One important advantage of this universal feature is that there is {\em no need to switch} between different controllers when the VSRs are requested to change their mode of operation---this is in contrast with other inner--loop schemes that require switchings between controllers, which is clearly undesirable in practice.
\endrem
%
\subsection{Relation of PI--PBC with Akagi's PQ method }
\label{akamet}
A dominant approach for the design of controllers for reactive power compensation using active filters (for three--phase circuits) is the PQ instantaneous power method proposed by Akagi, {\em et al.} in
\cite{akagi}. It consists of an outer--loop that generates references for the inner PI loops. The references are selected in order to satisfy a very simple heuristic: the AC active power $P$ has to be equal to
the DC power $P_{dc}$, thus ensuring the maximal power transfer from the AC to the DC side, and the reactive power should take a desired value. Now, using  \eqref{power} define the active AC and DC powers at the equilibrium as
$$
P^\star=v_di_d^\star,\quad P_{dc}^\star=v_C^\star i_{dc}.
$$
Consider then the following equivalences
\begin{equation*}
\begin{aligned}
P^\star P_{dc}= P_{dc}^\star P   \Leftrightarrow v_C^\star i_d = i_d^\star v_C \Leftrightarrow  y_1 =0,
\end{aligned}
\end{equation*}

with $y_1$ the first component of the passive output  \eqref{yy}. Similarly, for the reactive power
\begin{equation*}
\begin{aligned}
Q^\star P_{dc}= P_{dc}^\star Q  \Leftrightarrow v_C^\star i_q = i_q^\star v_C \Leftrightarrow  y_2 =0,
\end{aligned}
\end{equation*}

with $y_2$ the second component of the passive output  \eqref{yy}. In other words, the objective of the PI--PBC to drive the passive output $y$ to zero can be reinterpreted as a power equalization objective
identical to the one used in Akagi's PQ method.

%
\section{Performance Limitations of Inner--Loop PIs}
\label{zdsection}
%
Quality assessment of control algorithms is a difficult task---epitomized by the classical performance versus robustness tradeoff, neatly captured by the stability margins in linear designs. The situation for
nonlinear systems, where the notions of (dominant) poles and frequency response are specious, is far more complicated. In any case, it is well--known that the achievable performance in control systems is
limited by the presence of minimum phase zeros \cite{FRAZAM,QIUDAV,SERetal}.

In this section an attempt is made to evaluate the performance limitations of the inner--loop PI controllers discussed in the previous section. Towards this end, the zero dynamics of the VSR system
\eqref{sysiv} for the outputs $y$ \eqref{y}, $y_I$ \eqref{DCC} and $y_V$ \eqref{DOVC} are computed. All three outputs have relative degrees $\{1,1\}$, hence their zero dynamics is of order one but, while it is
exponentially stable for the passive output $y$ it turns out that---for normal operating regimes of the VSR---it is {\em unstable} for $y_I$ and $y_V$. If the {zero dynamics} is {\em unstable} cranking up the
controller gains yields an unstable behavior. This should be contrasted with the passive output $y$ that, as shown in Proposition \ref{pro2} yields an asymptotically stable closed--loop system for all positive
gains.\footnote{This discussion pertains only to the behavior of the adopted mathematical model of the VSR. In practice, other dynamical phenomena and unmodeled effects may trigger instability even for the
PI--PBC.}

To simplify the derivations only the case of $i_q^\star =0$ is considered. This assumption is justified since it corresponds to fix to zero the desired value of the reactive power, which is a common operating
mode of VSRs. Moreover, this is done without loss of generality because it is possible to show---alas, with messier calculations---that the stability of the zero dynamics is the same for the case of $i_q^\star
\neq 0$. This situation may arise when the VSR is associated to an AC grid and not to a renewable energy source. In this section, it is possible to prove that the (first order) zero dynamics associated to \eqref{y}, is ``extremely slow"---with respect to the overall bandwidth of the VSR. Since this zero ``attracts" one of the poles
of the closed--loop system, it stymies the achievement of fast transient responses. This situation motivates the inclusion of an outer--loop controller that generates the references to the inner--loop
PI. This modification is presented in Section \ref{sec6}.
%
\subsection{Zero dynamics analysis of the passive output $y$}
\label{zerdyn}
Before presenting the main result of this subsection, an important observation is done: the zero dynamics of the VSR model \eqref{sysiv} and of its corresponding incremental version are the same. Indeed,
the zero dynamics describes the behavior of the dynamical system restricted to the set where the output is zero. Since the incremental model dynamics is the {\em same} as the original model dynamics---simply
adding and substracting a constant---their zero dynamics coincide.

\begin{proposition}\em
\label{zeroprop}
Fix $(i^\star _{d},i^\star _{q},v^\star _{C}) \in {\mathcal{E}}$ with $i_q^\star =0$. The zero dynamics\footnote{With some abuse of notation, the zero dynamics is represented using the same symbols of the system
dynamics.} of the VSR \eqref{sysiv} with respect to the output \eqref{yy} is {\em exponentially stable} and is given by
\begin{equation}
\lab{zerdynpipbc}
\dot v_C=-\lambda v_C+\lambda v_C^\star ,\qquad \lambda:=\frac{R(i_d^\star )^2+G(v_C^\star )^2}{L(i_d^\star )^2+C(v_C^\star )^2}.
\end{equation}
\end{proposition}
\begin{proof}
By setting the output \eqref{yy} identically to zero and using  the fact that $i_q^\star=0$,  it is easy to get
\begin{equation}\label{yzero}
i_d=\frac{i^\star _d}{v_C^\star }v_C,\qquad i_q=\frac{i^\star _q}{v_C^\star }v_C=0.
\end{equation}
Replacing \eqref{yzero} into \eqref{sysiv} gives
\begin{align}
L\frac{i_d^\star }{v_C^\star }\dot v_C&=-R\frac{i_d^\star }{v_C^\star }v_C-v_C u_1 +v_d,\label{eq1}\\
0&=-L\omega\frac{i_d^\star }{v_C^\star }v_C-v_C u_2,\label{eq2}\\
C\dot v_C&=\frac{i_d^\star }{v_C^\star }v_C u_1 -Gv_C-i_{dc}.\label{eq3}
\end{align}
To eliminate $u_1$ it suffices to multiply \eqref{eq3} by $\frac{v_C^\star }{i_d^\star }$ and add it to  \eqref{eq1}, yielding
$$ 
\left(\frac{Cv_C^\star }{i_d^\star }+ \frac{L i_d^\star }{v_C^\star }\right)\dot v_C=-\left(\frac{Ri_d^\star }{v_C^\star }+ \frac{Gv_C^\star }{i_d^\star }\right)v_C+v_d-\frac{v_C^\star }{i_d^\star }i_{dc}. 
$$
The proof is completed by noting from \eqref{equset} that, for $(i^\star _{d},i^\star _{q},v^\star _{C}) \in {\mathcal{E}}$ with $i_q^\star =0$, it follows that
$$
v_d-\frac{v_C^\star }{i_d^\star }i_{dc}=\frac{R(i_d^\star )^2+G(v_C^\star )^2}{i_d^\star}
$$
and by pulling out the common factor $\frac{1 }{i_d^\star v_C^\star }$.
\end{proof}

\begrem
The parameters $R$ and $G$, that represent the losses in the VSR, are usually small---compared to $L$ and $C$. Consequently, $\lambda$ will also be a small value, placing the pole of the zero dynamics
very close to the origin and inducing slow convergence.
\endrem

\begin{remark}\em
It is interesting to note that the rate of exponential convergence of the zero dynamics can be rewritten as
$$
\lambda=\frac{1}{2}\frac{P^\star-P_{dc}^\star}{\mathcal{H}(i^\star _{d},i^\star _{q},v^\star _{C})},
$$

that is half the ratio between the steady--state dissipated power and the steady--state energy of the system. This relationship holds true also for the case $i_q^\star \neq 0$.
\end{remark}
%
\subsection{Zero dynamics analysis of $y_I$}
\label{zerdyn1}
Before analyzing the zero dynamics of the PQ and DC voltage control outputs,  \eqref{DCC} and  \eqref{DOVC}, respectively, it is important to recall that their references  do not necessarily  belong to the assignable
equilibrium set. However, the reasonable assumption that the zero dynamics admits an {\em equilibrium} for the chosen reference values can be done. If this is not the case the zero dynamics is unstable.
Moreover, similarly to the case of the passive output, it is assumed that $i_q^{\mathrm{ref}}=0$.

\begin{proposition}\em
\label{zeroprop1}
Fix $i_d^{\mathrm{ref}} \in \rea$, $i_q^{\mathrm{ref}}=0$. The zero dynamics of the VSR \eqref{sysiv} with respect to the output \eqref{DCC} is given by
\begin{equation}
\lab{zerdynDCC}
C \dot v_C=-G v_C+\frac{\alpha_I}{v_C} - i_{dc}^{\mathrm{ref}},\qquad \alpha_I:=v_di_d^{\mathrm{ref}}-R(i_d^{\mathrm{ref}})^2
\end{equation}
where $i_{dc}^{\mathrm{ref}}$ is a constant value for $i_{dc}$ satisfying
\begequ
\lab{equcon}
(i_{dc}^{\mathrm{ref}})^2> 4 G \alpha_I.
\endequ
\begite
\item[-] If $\alpha_I>0$ the zero dynamics has one equilibrium and it is {\em stable}.
\item[-] If $\alpha_I<0$ the zero dynamics has two equilibria one stable and one {unstable}.
\item[-] If $\alpha_I=0$ the zero dynamics is a linear asymptotically {\em stable} system.
\endite
\end{proposition}
\begin{proof}
Setting the output \eqref{DCC} equal to zero with $i_q^*=0$ and replacing into \eqref{sysiv} gives
\begin{align}
0&=-Ri_d^{\mathrm{ref}} -v_C u_1 +v_d\label{eq1DCC}\\
0&=-L\omega i_d^{\mathrm{ref}} -v_C u_2 \label{eq2DCC}\\
C\dot v_C&=i_d^{\mathrm{ref}}u_1 -Gv_C-i_{dc}^{\mathrm{ref}}\label{eq3DCC},
\end{align}
where the superscript $(\cdot)^{\mathrm{ref}}$ has been added to $i_{dc}$. Replacing $u_1$ obtained from \eqref{eq1DCC} into \eqref{eq3DCC} yields directly \eqref{zerdynDCC}. Condition \eqref{equcon} is then
necessary and sufficient for the existence of a (real) equilibrium of \eqref{zerdynDCC}. If $\alpha_I=0$ the dynamics reduces to $$C \dot v_C=-G v_C- i_{dc}^{\mathrm{ref}}.$$ The proof is completed by recalling that $v_C>0$ and looking at
the plots of  the right hand side of \eqref{zerdynDCC} for $\alpha_I$ positive and negative in Fig. \ref{figdcc}.
\end{proof}
\begin{figure*}[ht]
 \centering
\includegraphics[width=0.8\linewidth]{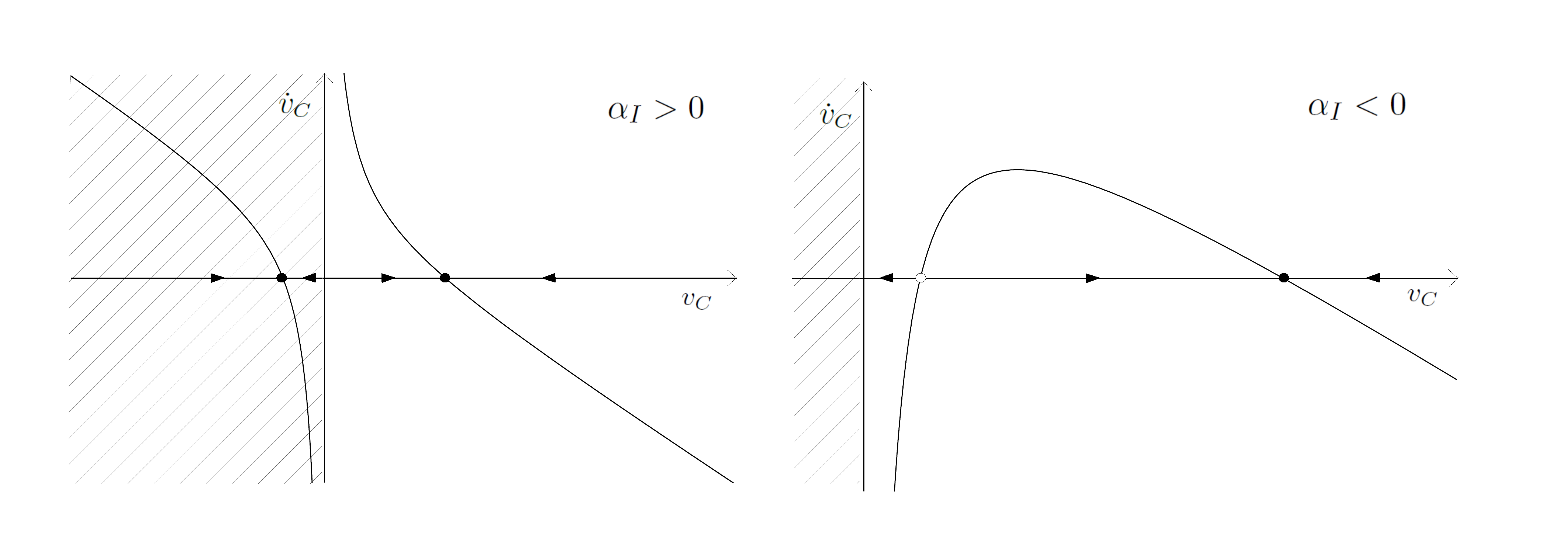}
 \caption{Plot of $\dot v_C$ versus $v_C$ for the cases of (a) $\alpha_I>0$ and (b) $\alpha_I<0$. The arrows in the horizontal axis indicate the direction of the flow of the zero dynamics.}
 \label{figdcc}
\end{figure*}

\begrem
From Fig. \ref{figdcc}, if $\alpha_I<0$, it is easy to see that the stable equilibrium point is the largest one. For standard values of the system parameters it turns out that this equilibrium is located beyond the physical operating regime of the system, hence it is of no practical interest.
\endrem

\begrem
\lab{rem5.6}
The parameters $R$ and $G$ are usually very small and $i_{dc}^\mathrm{ref}$ can take positive or negative values in standard operation. Then condition \eqref{equcon} is always verified while $\alpha_I$ can take positive or negative values.
\endrem

\begrem
\lab{rem5.7}
The situation $\alpha_I=0$, when the zero dynamics is linear and asymptotically stable, is unattainable in applications. Indeed, assuming that in steady--state all signals converge to their reference values, it can be shown that $\alpha_I=0$ if and only if $Gv_C+i_{dc}=0$ that, given the small values of $G$ is not realistic in practice.
\endrem
%
\subsection{Zero dynamics analysis of $y_V$}
\label{zerdyn2}

\begin{proposition}\em
\label{zeroprop2}
Fix $v_C^{\mathrm{ref}} \in \rea$, $i_q^{\mathrm{ref}}=0$. The zero dynamics of the VSR \eqref{sysiv} with respect to the output \eqref{DOVC} is given by
\begin{equation}
\lab{zerdynDOVC}
L \frac{d i_d}{dt}=-R i_d-\frac{\alpha_V}{i_d} + v_d,\qquad \alpha_V:=i_{dc}^{\mathrm{ref}}v_C^{\mathrm{ref}}+G(v_C^{\mathrm{ref}})^2
\end{equation}
where $i_{dc}^{\mathrm{ref}}$ is a constant value for $i_{dc}$ satisfying
\begequ
\lab{equcon1}
v_d^2>-4 R \alpha_V.
\endequ
\begite
\item[-] If $\alpha_V<0$ the zero dynamics has two equilibria and they are both {\em stable}.
\item[-] If $\alpha_V>0$ the zero dynamics has two equilibria one stable and one {unstable}.
\item[-] If $\alpha_V=0$ the zero dynamics is a linear asymptotically {\em stable} system.
\endite
\end{proposition}
\begin{proof}
Setting the output \eqref{DOVC} equal to zero with $i_q^*=0$ and replacing into \eqref{sysiv} gives
\begin{align}
L \frac{di_d}{dt}&=-Ri_d -v_C^{\mathrm{ref}} u_1 +v_d,\label{eq1DOVC}\\
0&=-L\omega i_d -v_C^{\mathrm{ref}} u_2 ,\label{eq2DOVC}\\
0&=i_d u_1 -Gv_C^{\mathrm{ref}}-i_{dc}\label{eq3DOVC}.
\end{align}
Replacing $u_1$ obtained from \eqref{eq3DOVC} into \eqref{eq1DOVC} yields directly \eqref{zerdynDOVC}. Condition \eqref{equcon1} is necessary and sufficient for the existence of a (real) equilibrium of
\eqref{zerdynDOVC}. The proof is completed invoking the same arguments used in the proof of Proposition \ref{zeroprop1} and are omitted for brevity.
\end{proof}

\begrem
Remarks \ref{rem5.6} and \ref{rem5.7} apply \textit{verbatim} to \eqref{zerdynDOVC} and $\alpha_V$ of Proposition \ref{zeroprop2}.
\endrem

\subsection{Simulated evidence of the performance limitations}
\label{simuinner}
Although Proposition \ref{zeroprop} proves that the zero dynamics for the passive output $y$ is exponentially stable, it turns out that, for the components used in standard HVDC transmission system, the convergence rate is $\lambda \approx 0.04$, which is extremely slow. As indicated above this dominating dynamics stymies the achievement of fast transient responses---a situation that is shown in the following simulations. Also, simulated evidence of the unstable behavior of the PI inner--loops using the outputs  \eqref{DCC} and \eqref{DOVC} is presented.\\
A three--terminals HVDC transmission system with a simple \textit{meshed} topology is considered, as illustrated in Fig. \ref{tops}, where the corresponding graph is also given. The model of the system is given by \eqref{overall}, that is a system of dimension $3n+\ell=11$ with $2n=6$ inputs. Parameters of the VSRs and of the transmission lines are given in Table 1.

\begin{figure}[h!]
 \centering
\includegraphics[width=1\linewidth]{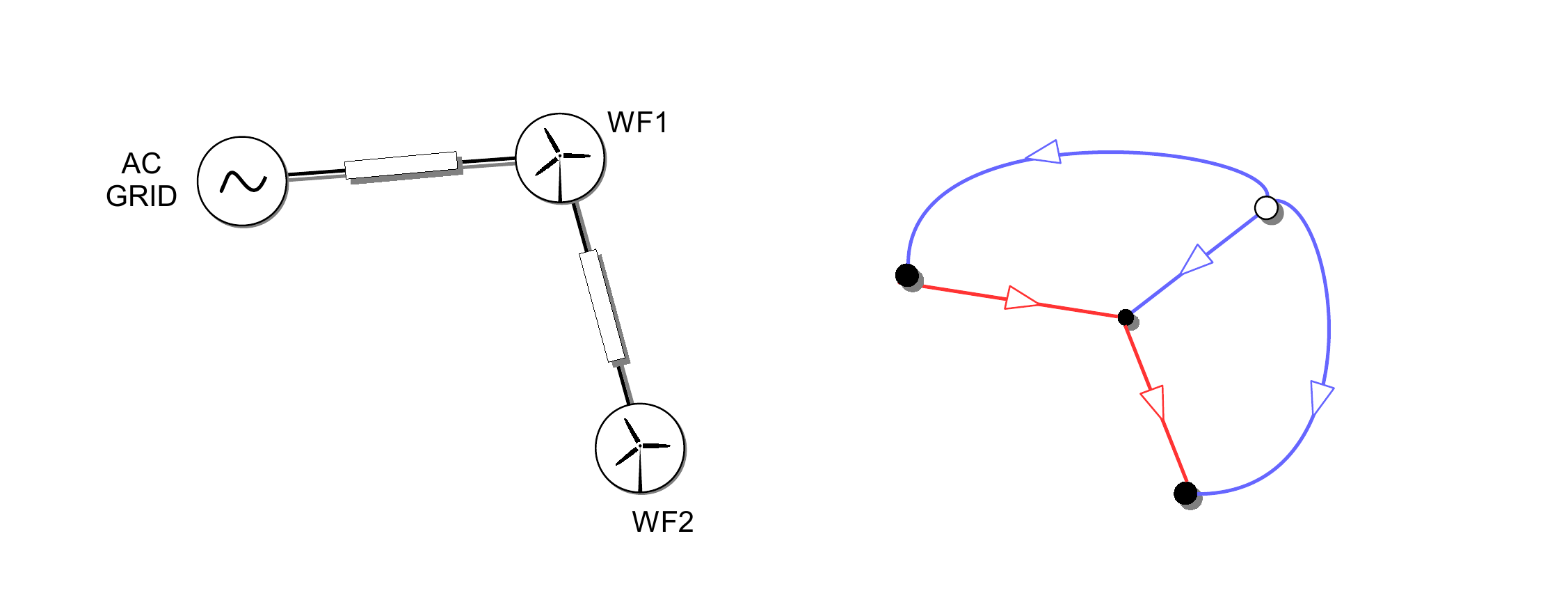}
 \caption{{Schematic representation of an HVDC transmission system constituted by three stations, associated to two wind farms (WFs) and an AC grid, with associated graph. The graph is represented by filled circles for the \textit{VSRs-buses} and the unfilled circle for the \textit{ground} node. Blue and red edges characterize VSRs and lines, respectively.}}
 \label{tops}
\end{figure}

Consider then the following {\em control objectives}: all the stations are required to regulate the reactive power to zero; the stations associated to the wind farms (WF1, WF2) are required to regulate the active power to desired (constant) values; the remaining station, called \textit{slack bus} (SB), must regulate the voltage around its nominal value. In Table 2, the corresponding references of direct current and DC voltages are furnished, together with the corresponding assignable equilibria, that are calculated via the PFSSE defined by \eqref{prop}. Changes in references occur every $T$ $s$ over a time interval of $5T$ $s$. It should be noticed that from $0$ to $2T$ the power flow is uniquely directed from both wind farms stations to the AC grid, while at $2T$  and next $3T$ the wind farms stations start demanding power to the AC grid, thus reversing the direction of the power flow. This situation can arise when the power produced by the wind farms is insufficient to supply local loads.
\begin{table}
\caption{System parameters.}
\label{parameters1}       
%
%
\centering
\begin{tabular}{p{1cm}p{1.35cm}p{1cm}p{1.35cm}}
\hline
  & Value &   & Value \\
\hline
 $R_{r,i}$               & $0.01$ $\Omega $ & $G_{r,i}$  & $0$ $\Omega^{-1} $ \\
         $L_{r,i}$               & $40$ $mH$  & $C_{r,i}$ & $20$ $\mu F$ \\
         $V_i$                   & $130$ $kV$ &  $\omega_i$  & $50$ $Hz$ \\
         $R_{\ell,12}$           & $26$ $\Omega$ & $L_{\ell,12}$&$3.76$ $mH$\\
         $R_{\ell,23}$           & $20$ $\Omega$ & $L_{\ell,23}$&$2.54$ $mH$\\
\hline 
\end{tabular}
\end{table}
\begin{table}
\caption{System references.}
\label{desired}       
%
%
\begin{tabular}{p{0.3cm}p{0.8cm}p{0.8cm}p{0.8cm}p{0.7cm}p{1cm}p{1cm}}
\hline\noalign{\smallskip}
          & $\;\;\quad SB$    &$\;\;{WF}_1$    &$\;\;{WF}_2$  &$ \;SB$ &$\;\;\;{WF}_1$ &$\;\;\;{WF}_2$\\
\hline
        $0$   &   $-1260$  & $\;\;\;\; 900$ &  $\;\; 1000$ & $100$ & $142.595$  & $158.951$ \\
        $T$ &    $-1588$ & $\;\;\;\;900$  & $\;\; 1800$ &  $100$& $153.650$  & $179.691$  \\
        $2T$&    $\;\;\;-266$ & $\;\;\;\;500$  & $-200$ & $100$ &$109.004$  & $104.004$  \\
        $3T$&    $\;\;\;\;\;\;\;\; 905$ & $-400$  & $-200$ & $100$ &$\;\;\; 69.419$  & $\;\;\;60.877$  \\
        $4T$&    $\;\;\;-849$ & $\;1300$  & $-200$ & $100$ &$128.708$  & $124.532$  \\
\hline
\end{tabular}
\end{table}

\subsubsection{PI--PBC}
\label{simpipbc}
In this subsection the simulations on the three--terminals benchmark example of the decentralized PI--PBC defined in Subsection \ref{PBCcontrol} are presented, illustrating the stability properties and performance limitations previously discussed. Setting $T=2000$ $s$ the controllers \eqref{control} are designed with identical parameters and diagonal matrices $k_{P,i}=\mathrm{diag}\{1,1\}$, $k_{I,i}=\mathrm{diag}\{10,10\}$. The behavior of the VSRs are depicted in Fig. \ref{pbcfig}.
\begin{figure*}[ht]
 \centering
 \includegraphics[width=0.8\linewidth]{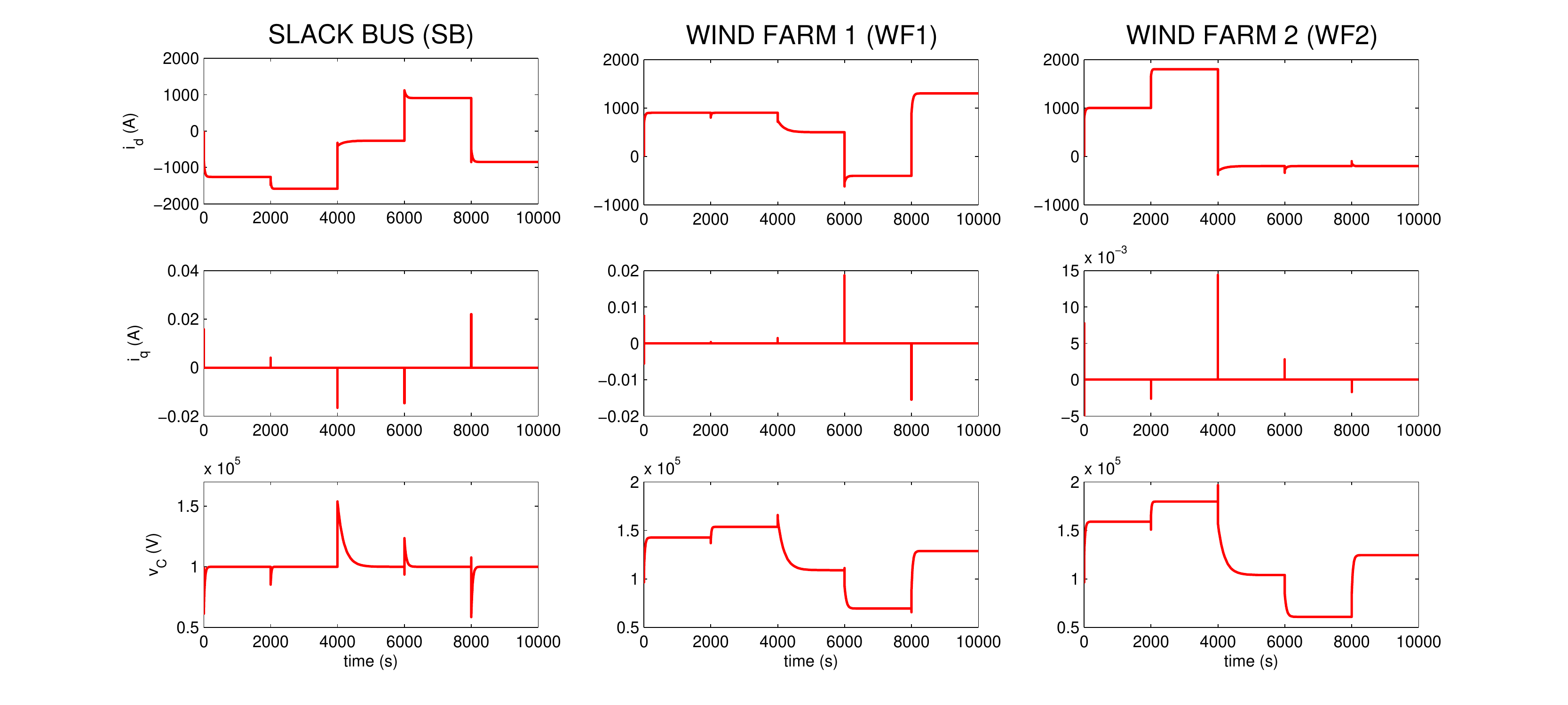}
 \caption{{Responses of VSRs variables under the decentralized PI--PBC.}}
 \label{pbcfig}
\end{figure*}

As expected, the direct currents of each station attain the assignable equilibria defined in Table 2, while the quadrature currents are always kept to zero after a very short transient. Moreover, the DC voltage at the slack bus is maintained near the nominal value of $100$ $kV$, as required, while the DC voltage variation at the wind farms stations, balances the fluctuation of power demand. Even though the desired steady--state is attained for all practical purposes, the convergence time of direct currents and DC voltages is extremely slow. This poor transient performance behavior is independent of the controller gains. Indeed, extensive simulations show that the system maintains the same slow convergence time even with larger gains, thus validating the performance limitations analysis realized in subsection \ref{zerdyn}.
%

\subsubsection{PQ and DC voltage controllers}

\begin{figure*}[ht]
 \centering
 \includegraphics[width=0.8\linewidth]{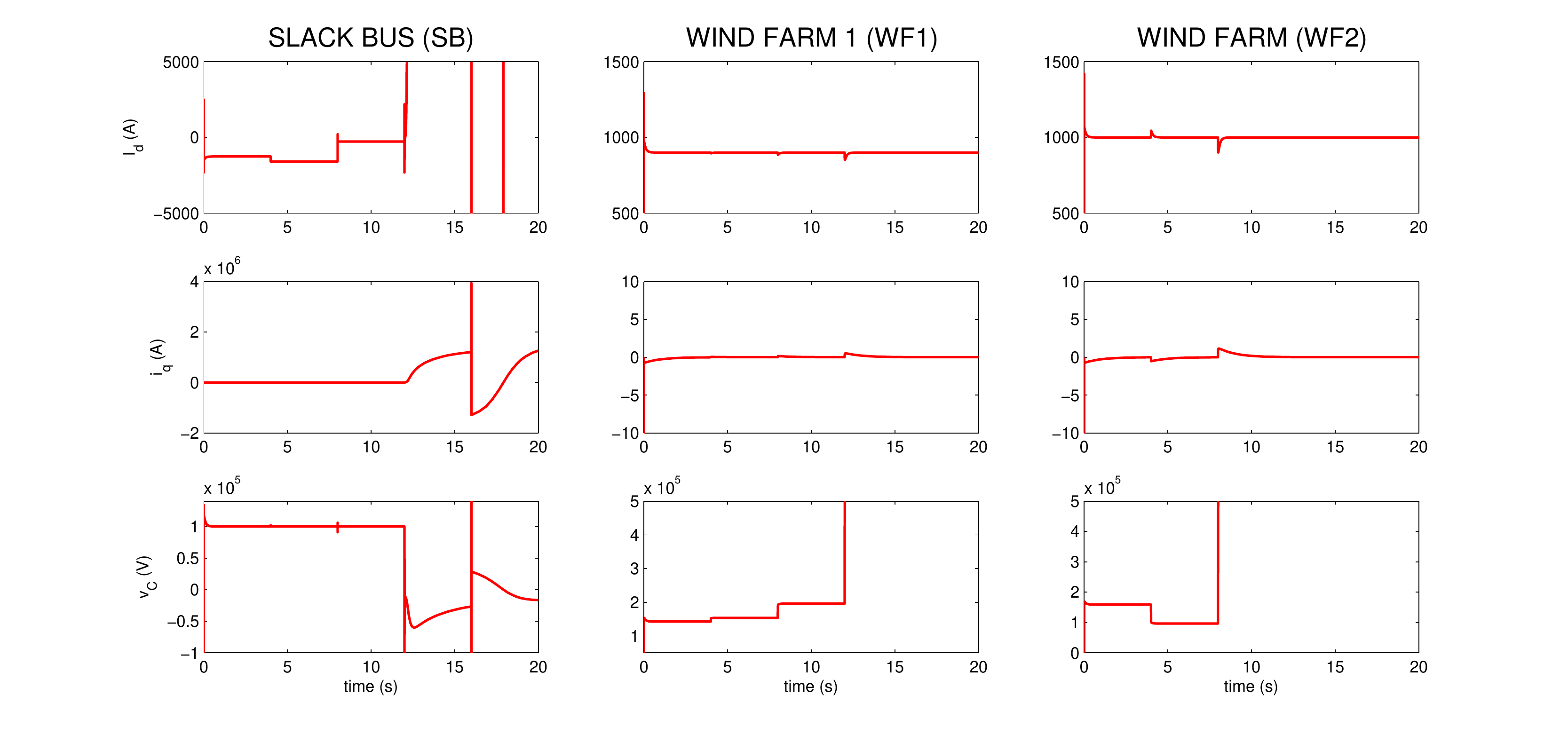}
 \caption{{Responses of VSRs variables under the decentralized PQ and DC voltage controllers.}}
 \label{PIfig}
\end{figure*}

The behavior of the system under the standard PQ and DC voltage controllers of Subsection \ref{liter} is next analyzed. In agreement with the control requirements described above, two PQ controllers are designed to regulate direct and quadrature currents of the wind farm  stations  and one DC voltage controller is designed to regulate DC voltage and quadrature current of the slack bus. Simple PI controllers defined over the outputs \eqref{DCC}, \eqref{DOVC} are considered, designed with identical gains $k_{P,i}$, $k_{I,i}$ tuned via simulations. The behavior of the VSRs are depicted in Fig. \ref{PIfig}, with $T=4$ $s$. This value should be contrasted with the value ($T=2000$ $s$) used for the PI--PBC. It is easy to see that the PQ and DC voltage controllers correctly (and rapidly) regulate the station at the desired references between $0$ and $8$ $s$. This good behavior is not surprising, because PQ controllers applied to VSRs that are injecting power and a DC voltage controller applied to VSRs that is absorbing power, have associated globally asymptotically stable zero dynamics, as proved in Subsections \ref{zerdyn1}, \ref{zerdyn2}. On the other hand, as shown in the figures, when at stations $WF1$ and $WF2$ the power flow is reversed (respectively at $t=12$ $s$ and $t=8$ $s$), the correspondent DC voltages go unstable, because in these cases the zero dynamics is unstable. Similar unstable behavior appears also at the slack bus station.

\section{Adding an outer--loop to the PI--PBC}
\lab{sec6}
%
To overcome the transient performance limitations of the PI--PBC exhibited in Subsection \ref{simpipbc}, in this section it is proposed to add an outer--loop that takes as input some desired references---indicated with $(\cdot)^{\mathrm{ref}}$---and generates as output the references to the inner--loop scheme---see Fig. \ref{droopfg}. The latter will replace the desired equilibria in the definition of the passive output \eqref{yy}, associated to each VSR. Because in this paper the discussion is restricted to the performances of the PI--PBC in nominal operating conditions, the following assumption is made.
\smallbreak

\begin{assumption}\label{assref} The input references $(\cdot)^{\mathrm{ref}}$ of the outer--loop control belong to the set of assignable equilibria $\mathcal{E}$.
\end{assumption}
\smallbreak

In common practice, the references taken as input by the outer--loop control may not belong to the assignable set, thus posing the more general problem of designing an outer--loop that ensures convergence to an (a priori unknown) equilibrium point, \textit{i.e.} the so-called \textit{primary control} problem \cite{berteen,dorfler,sandberg}. However, in the present case the attention is exclusively focused on the aspect of overcoming of performance limitations and  a robustness analysis is left for future investigation.

\subsection{New output with droop control}
A commonly used outer--loop control is the so--called {\em conventional droop control}, which replaces---at the $i$--th VSR---the direct current $i_{d,i}^\star $ with its desired reference $i_{d,i}^{\mathrm{ref}}$ plus a deviation (droop) term proportional to the voltage error, leaving some constant references for  $i_{q,i}^\star $ and  $v_{C,i}^\star $. More precisely, the following assignments are made in  \eqref{yy}
\begin{equation}
\label{droop}
i_{d,i}^\star  \leftarrow i_{d,i}^{\mathrm{\mathrm{ref}}}+k_{d,i}(v_{C,i}^{\mathrm{ref}}-v_{C,i}),\quad i_{q,i}^\star  \leftarrow i_{q,i}^{\mathrm{ref}},\quad v_{C,i}^\star \leftarrow v_{C,i}^{\mathrm{ref}},
\end{equation}
where $k_{d,i}>0$ is called the droop coefficient. 
 Replacing \eqref{droop} in the passive output \eqref{yy} yields the new output
\begin{equation}
\label{extendeddroop}
y_{N,i}:=\lef[{c} v_{C,i}^{\mathrm{ref}}i_{d,i}-i_{d,i}^{\mathrm{ref}}v_{C,i}-k_{d,i}(v_{C,i}^{\mathrm{ref}}-v_{C,i})v_{C,i}\\ v_{C,i}^{\mathrm{ref}}i_{q,i}-i_{q,i}^{\mathrm{ref}}v_{C,i}\rig].
\end{equation}

\begin{figure*}[ht]
 \centering
 \includegraphics[width=0.8\linewidth]{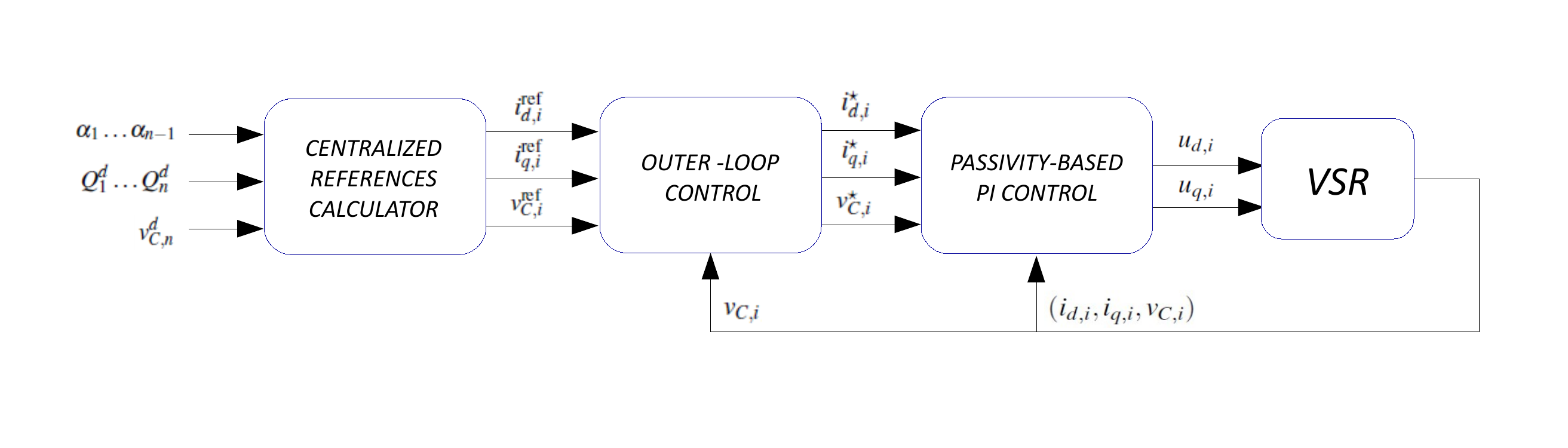}
 \caption{{Control architecture of HVDC transmission systems under nominal operating conditions.}}
 \label{droopfg}
\end{figure*}
The simulations show that the performance of the modified PI--PBC, that is, adding a PI around the new outputs \eqref{extendeddroop}, is significantly better than the original PI--PBC and that, under Assumption \ref{assref}, the same closed--loop equilibrium point is preserved. 


\begin{remark}\em
In several papers, {\em e.g.},  \cite{sandberg,bencha}, the assumption that the inner--loop PI and the VSRs are much faster than the network dynamics is made. Under this assumption, a globally asymptotically stabilized VSR can be equivalently modeled as the parallel interconnection of two AC current sources connected to the network through a voltage bus capacitor, that is assumed to operate at the same time scale of the network. The resulting reduced model is linear and conditions for {\em local} asymptotic stability can be easily established \cite{sandberg,dorfler}.
\end{remark}
\begrem
In contrast to \cite{sandberg,sun}, the droop control laws are defined with respect to DC voltages that do not necessarily all coincide with the same nominal value $v_C^{nom}$. Then, the trade--off between having the voltages converge to the nominal voltage, and satisfying a pre--determined active power distribution (\textit{sharing}) between the VSRs, is completely captured by the power flow steady-state equations that determine the set of assignable equilibria. The reader is referred to Section \ref{sec7} for further details on this problem.
\endrem

\subsection{A GAS outer--loop controller}
It is evident from \eqref{extendeddroop}, that the inclusion of additional state--dependent terms to the first component of $y_{N,i}$ invalidates the stability result obtained in Proposition \ref{pro2}, as the new output $y_{N,i}$ is {\em not passive}. Moving from these considerations, a modified version of the PI--PBC \eqref{control} is proposed. The latter, if properly designed, allows to overcome the performance  limitations of the passive output, while preserving global asymptotic stability of the closed--loop system. This modification consists of an additional linear feedback that affects only the proportional part of the PI--PBC.  The following  proposition is then presented.\\
\begin{proposition}\em
\lab{prolast}
Consider the HVDC transmission system \eqref{overall}, with a desired steady--state $x^\star  \in\mathcal{E}$, in closed--loop with the PI control
\begin{equation}\label{control2}
\begin{aligned}
u & = -K_P y - K_I\zeta-K_L Q\tilde x,\quad \dot \zeta = y,
\end{aligned}
\end{equation}
with $y$ given in \eqref{y}, gain matrices $K_P, K_I$ as in \eqref{gains} and \mbox{$K_L\in\mathbb{R}^{2n\times(3n+\ell)}$} verifying
\begin{equation}\label{condlin}
\mathcal{R}_0:=\mathcal{R}+g(x^\star) K_Pg^\top(x^\star)+\frac{1}{2}\left[g(x^\star)K_L+K_L^\top g^\top(x^\star)\right]>0.
\end{equation}
 Then, the equilibrium point $(x^\star ,K_I^{-1}u^\star )$ is globally asymptotically stable (GAS).
\end{proposition}
\begin{proof}
Using the same Lyapunov function \eqref{weq} employed in the proof of Proposition \ref{pro2}, the derivative along the trajectories of the closed--loop system \eqref{overall}--\eqref{control2} is given by
\begin{equation*}
\begin{aligned}
\dot{{W}}&=-\tilde x^\top Q\mathcal{R}Q\tilde x+{y}^\top\tilde{u}+\tilde\zeta^\top K_Iy\\
&=-\tilde x^\top Q\mathcal{R}Q\tilde x+{y}^\top\tilde{u}-(\tilde u^\top +y^\top K_P+\tilde x^\top QK_L^\top)y  \\
&=-\tilde x^\top Q\mathcal{R}Q\tilde x-\tilde x^\top Qg(x^\star)K_Pg^\top(x^\star)Q\tilde x-\tilde x^\top QK_L^\top g^\top(x^\star)Q\tilde x\\
& = -\tilde x^\top Q\mathcal{R}_0Q\tilde x<0
,
\end{aligned}
\end{equation*}

where in the third equivalence the output definition \mbox{$y=g^\top(x^\star)Q\tilde x$} is used, while the last equivalence follows from condition \eqref{condlin}.
\end{proof}
\smallbreak

Loosely speaking, the Proposition \ref{control2} states that the property of global asymptotic stability of the closed--loop system \eqref{overall}--\eqref{control}---that is the HVDC transmission system controlled via PI--PBC---is preserved for any additional linear feedback that affects only the proportional part of the controller and any gain matrix $K_L$ which verifies condition \eqref{condlin}. However, beside this stability result, Proposition \ref{prolast} does not provide any hint on how to select the controller gains in order to overcome the performance limitations of the PI--PBC, nor how to preserve the decentralization property that -- for some inappropriate choice of the gain matrix -- can be even lost. Taking inspiration from the conventional droop controller discussed in the previous section, the following assignment is made:
$$
K_L:=\begin{bmatrix}
0&0&K_D&0\\
0&0&0&0
\end{bmatrix},
$$
where $K_D:=\mathrm{diag}\{k_{D,i}\}\in\mathbb{R}^{n\times n}$ is a positive matrix to be defined. With this choice it is easy to see that the controller \eqref{control2} can be decomposed in $n$ decentralized controllers of the form
\begin{equation}\label{decpi}\resizebox{1\hsize}{!}{$
\begin{aligned}
\begin{bmatrix}
u_{d,i}\\u_{q,i}
\end{bmatrix}=\begin{bmatrix}
k_{Pd,i}y_{d,i}-k_{Id,i}z_{d,i}-k_{D,i}(v_{C,i}-v_{C,i}^\star)\\
k_{Pq,i}y_{q,i}-k_{Iq,i}z_{q,i}.
\end{bmatrix},\quad \dot{\begin{bmatrix}
z_{d,i}\\
z_{q,i}
\end{bmatrix}}&=\begin{bmatrix}
y_{d,i}\\
y_{q,i}
\end{bmatrix},
\end{aligned}$}
\end{equation}

that correspond to $n$ PI--PBC plus an additional linear feedback in the local DC voltage error. Straightforward calculations---here omitted for brevity---show that it is always possible to determine a gain  matrix $K_D$, such that \eqref{condlin} is verified, thus guaranteeing global asymptotic stability of the closed--loop system.
\begrem
The modified PI--PBC \eqref{decpi} can be interpreted, similarly to the droop controller \eqref{droop}, as an outer--loop  providing references for the the standard PI--PBC, \textit{but only} affecting its proportional part. It is indeed easy to see that it corresponds to assume the following inner--loop control scheme
\begin{align*}\resizebox{1\hsize}{!}{$
\begin{bmatrix}
u_{d,i}\\u_{q,i}
\end{bmatrix} =\begin{bmatrix}
k_{Pd,i}(v_{C,i}^{\star P} i_{d,i}-i_{d,i}^{\star P} v_{C,i})-k_{Id,i}z_{d,i}\\
k_{Pq,i}(v_{C,i}^{\star P} i_{q,i}-i_{q,i}^{\star P} v_{C,i})-k_{Iq,i}z_{q,i}
\end{bmatrix},\;\;
\dot{\begin{bmatrix}
z_{d,i}\\
z_{q,i}
\end{bmatrix}}=\begin{bmatrix}
v_{C,i}^{\star I} i_{d,i}-i_{d,i}^{\star I} v_{C,i}\\
v_{C,i}^{\star I} i_{q,i}-i_{q,i}^{\star I} v_{C,i}
\end{bmatrix},$}
\end{align*}
together with the following outer--loop assignments of the proportional and integral references
\begin{equation}\label{drooppi}
\begin{aligned}
&i_{d,i}^{\star,P}  \leftarrow i_{d,i}^{\mathrm{\mathrm{ref}}}+k_{D,i}\frac{v_{C,i}-v_{C,i}^{\mathrm{ref}}}{v_{C,i}},\quad i_{q,i}^{\star,P}  \leftarrow i_{q,i}^{\mathrm{ref}},\quad v_{C,i}^{\star,P} \leftarrow v_{C,i}^{\mathrm{ref}}\\
&i_{d,i}^{\star,I}  \leftarrow i_{d,i}^{\mathrm{ref}},\quad i_{q,i}^{\star,I}  \leftarrow i_{q,i}^{\mathrm{ref}},\quad v_{C,i}^{\star,I} \leftarrow v_{C,i}^{\mathrm{ref}},
\end{aligned}
\end{equation}
where, as done before, the notation $(\cdot)^{\mathrm{ref}}$ indicates the (assignable) references of the outer--loop.
\endrem

\subsection{Simulations}
\label{simudroop}

To illustrate the previous discussion on outer--loop controllers, the three--terminals benchmark example described in Subsection \ref{simuinner}---controlled via decentralized PI--PBC---is considered. The same control parameters of Subsection \ref{simpipbc} are employed, and the benefits in terms of performance, provided by adding an outer--loop control to the PI--PBC controllers of the form \eqref{drooppi}, are further analyzed. For the choice of the controller gains a very simple heuristic, often invoked in conventional droop control, is selected. Because droop coefficients are supposed to quantify the  additional dissipation injected into the voltage dynamics, and because the rate of convergence of the same depend from the value of the capacitances, define
\begin{equation}\label{Kdroop}
d_i=\frac{G_i+k_{D,i}}{C_i},\qquad i\in[1,n]
\end{equation}
as a measure of the convergence rate of the $i$-th station, with $G_{i}$ and $k_{D,i}$ the conductance and the droop coefficient of the VSR, respectively. A possible choice of droop coefficients is to define a common convergence rate $d$ such that $d_i=d$ for every $i$, that is equivalent to define an uniform convergence rate over the three stations. In the three--terminals benchmark example, because parameters are supposed to be identical at each VSR, the droop coefficients will take identical values, namely $k_{D,i}=5\cdot 10^{-2}$.
\begin{figure*}[ht]
 \centering
 \includegraphics[width=0.8\linewidth]{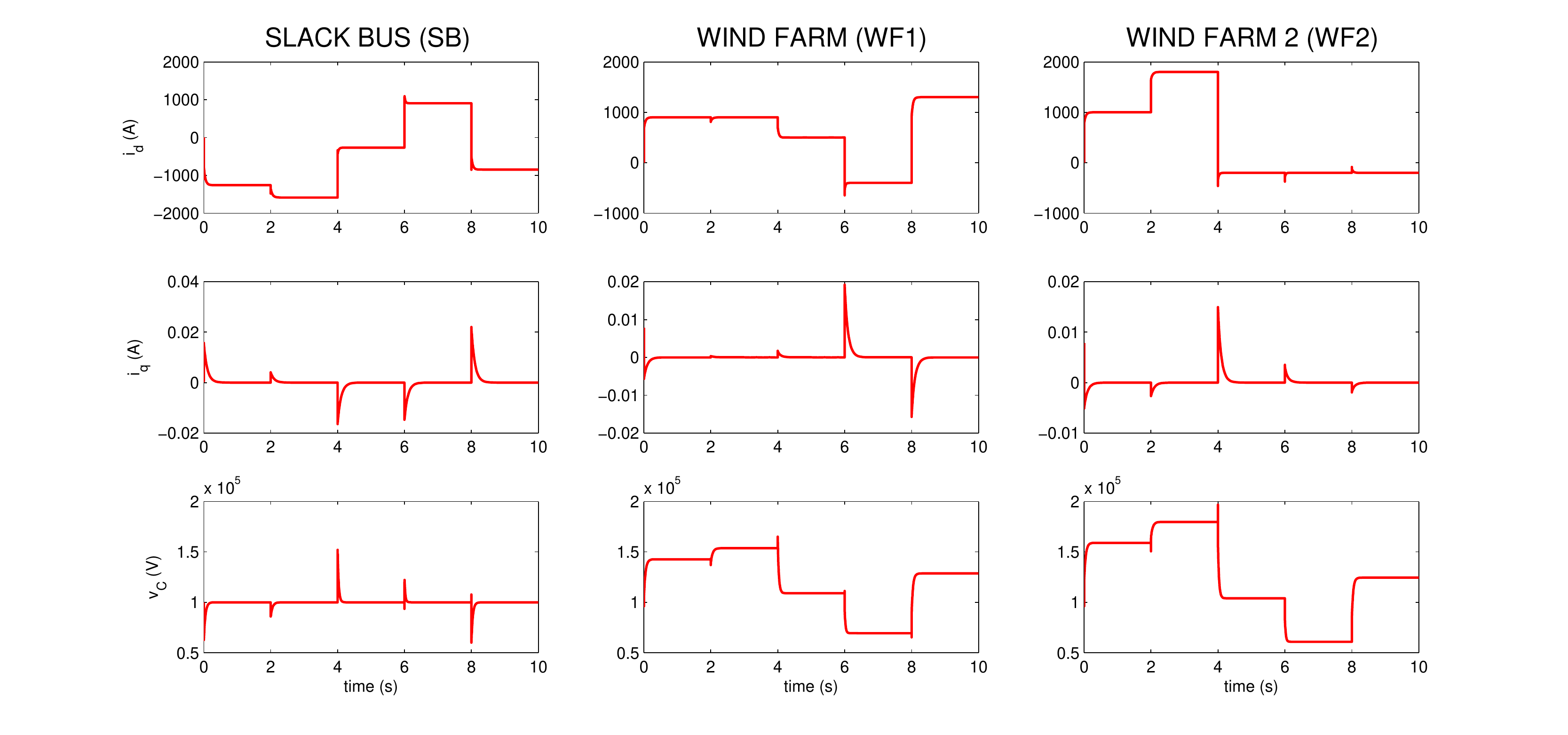}
 \caption{{Responses of VSRs variables with the decentralized PI--PBC plus GAS outer controller.}}
 \label{droopfig}
\end{figure*}
The behavior of the VSRs are illustrated in Fig. \ref{droopfig}. In contrast to the simulations of the basic PI--PBC of Subsection \ref{simpipbc} when the references change every $T=2000\;s$, now they are a {\em thousand times faster} that is, every $T=2$ $s$. It is easy to see that, compared to Fig. \ref{pbcfig}, the responses maintain the same shape while the convergence occurs with a rate $\approx 10^3$ faster.

\begrem
Under Assumption \ref{assref}, the responses of the VSRs under the PI--PBC plus conventional droop control---here omitted for brevity---are very similar to the responses of the GAS outer controller, depicted in Fig. \ref{droopfig}. However, if Assumption \ref{assref} is not verified, \textit{e.g.} in perturbed operating conditions, significant differences occurs between the VSRs responses. It is indeed possible to verify that, while the conventional droop control ensures convergence to an (assignable) equilibrium point independently from the assigned references, the GAS outer controller may experience instability.
\endrem

%
\section{Centralized References Calculator}
\label{sec7}
%
In this section, a reformulation of the problem of choosing appropriate references for the droop or inner PI controllers presented in the previous sections is provided. In power systems literature, this is
often referred as \textit{references calculator}  \cite{berteen}  and it is in general characterized by a centralized architecture ---see Fig. \ref{droopfg}. 

It is next shown how certain simple constrains---widely employed for the references calculation---can be mathematically formalized using the PFSSE. It is worth mentioning that the constraints adopted here, are only special cases of a more complex optimization problem, that in general requires to take into account many other aspects, related to technical and economical issues. However, the following analysis is limited to some of the most relevant technical aspects, leaving a more diverse investigation as a future work.

Consider an HVDC transmission system with \textit{meshed} topology composed by $n$ VSRs (stations), and described by the pH system \eqref{overall}.  The set of assignable references is determined by the following PFSSE
\begin{equation}\label{secondary}
\begin{split}
- R_i(i_{d,i}^{\mathrm{ref}})^2&-R_i(i_{q,i}^{\mathrm{ref}})^2 -G_i(v_{C,i}^{\mathrm{ref}})^2+\\
&+v_{d,i}i^{\mathrm{ref}}_{d,i}-v^{\mathrm{ref}}_{C,i}M_{i}\mathcal{R}_L^{-1}M_{i}^\top v^{\mathrm{ref}}_C=0,
\end{split}
\end{equation}
for $i\in[1,n]$, where the row vector $M_{ i}\in\mathbb{R}^n$ is the $i$--th row of the matrix $M $, that coincide with \eqref{equ2}, but in co-energy variables. The PFSSE consist in $n$ quadratic, coupled equations in $3n$ variables, one for each station. A possibility is then to directly assign $2n$ variables, that correspond to the desired references, while the remaining $n$ variables can be easily determined via \eqref{secondary} and then be provided to the inner--loop controllers. However, the choice of which variables have to be chosen as desired references is not arbitrary, but depends on the control objectives to satisfy. It is next illustrated how the problem of defining references in conformity with some natural control objectives can be formalized using the PFSSE. \\
Assume the following requirements for the $n$ stations:  keep the DC voltage of only one station (called \textit{slack bus}) close to the nominal value, guarantee a proportional active power distribution (\textit{power sharing}) among the stations, regulate the reactive power to a desired value at each station. These requirements can be easily reformulated as constraints over the PFSSE as follows:

\begin{itemize}
\item[-] regulation of the DC voltage of the \textit{slack bus}
\begin{equation*}\label{DCV}
v^{\mathrm{ref}}_{C,n}=v^{\mathrm{d}}_{C,n},
\end{equation*}
where $v_{C,n}^{\mathrm{d}}$ represents the DC voltage nominal value;
\item[-] proportional \textit{power sharing}
\begin{equation*}\label{Ps2}
P_{i}^{\mathrm{ref}}=\alpha _{i} P_n^{\mathrm{ref}}\quad\Rightarrow\quad i^{\mathrm{ref}}_{d,i}=\left[\frac{v_{d,n}}{v_{d,i}}\alpha_{i}\right] i^{\mathrm{ref}}_{d,n},\qquad i\in[1,n-1]
\end{equation*}
where $\alpha_{i}$ is a ratio that determines the proportional active power distribution of the $i$-th station with respect to the \textit{slack bus};
\item[-] regulation of the reactive power
\begin{equation*}\label{react}
Q_i^{\mathrm{ref}}=Q_i^\mathrm{d}\quad\Rightarrow\quad i^{\mathrm{ref}}_{q,i}=\frac{Q^{d}_i}{v_{d,i}},\qquad i\in[1,n]
\end{equation*}
where $Q^\mathrm{d}_i$ represents the exact reactive power required to be injected (or absorbed) by the $i$-th station.
\end{itemize}

It is easy to see that the equations above constitute indeed a set of $2n$ assignments for the PFSSE, that can be consequently solved with respect to the remaining variables.

\section{Conclusions and Future Perspectives}
\lab{sec8}
%
The present work covers different aspects of modeling, analysis and control of multi--terminal HVDC transmission systems. The main contribution is a decentralized, globally asymptotically stable, PI control for a very general class of  multi--terminal HVDC transmission systems. For this purpose, starting from a graph description of the network, a pH representation has been obtained, thus revealing the intrinsic passivity properties of the system. The result is a direct extension of the previous works on PI control of VSRs, to a sufficiently general interconnected system, with the important property that the control is decentralized, a fundamental requirement for large--scale systems. To provide some connections between the proposed controller and standard techniques, widely used in literature, a comparative analysis of stability and performances is provided, shedding some light on limitations and benefits of different approaches. In particular it is proved---and validated via simulations---that the popular current and voltage control techniques possibly lead to unstable behaviors of the controlled system, while the proposed PI--PBC, although ensuring convergence, has clear performance limitations. The theoretical analysis that substantiates these claims is based on a detailed, nonlinear zero dynamics analysis of a single VSR with respect to the outputs used for all these controllers. To overcome the performance limitations of the PI--PBC,  an outer--loop controller is added. This outer--loop takes the form of a voltage droop---that is the {\em de facto} standard in the power systems community. Taking inspiration from its usual formulation  an alternative controller that  overcomes the performance limitation of the PI--PBC is then developed, further showing that global asymptotic stability is preserved. 

A future research line pertains to the use of more accurate models for the description of the system, that may improve the control quality. For instance, the behavior of long transmission lines is best described by means of the Telegrapher's equations, thus leading to an infinite dimensional pH representation, which can still be handled with existing theory \cite{jeltsema}. Because the conventional droop controller destroys the passivity property---that is instrumental for the stability analysis of the PI--PBC---current research is  under way to establish some stability properties of the PI--PBC plus conventional droop control. A further possibility is the development of new provably stable outer--loop primary controllers, that ensure convergence to an (a priori unknown) equilibrium point, while the references do not belong to the set of assignable equilibria.  Although the latter approach is of theoretical interest, it is the authors' belief that a rigorous analysis of the droop control---to substantiate its widely acknowledged robustification features---would better contribute to bridge the gap between theory and engineering practice. It would also be of interest to investigate in detail new strategies for the references calculation, moving away from the PFSSE. A viable possibility is to consider the latter as a problem of static optimization, that allows a simple characterization of the control objectives. A final, long term, objective is the experimental validation of the proposed PI--PBC plus droop control scheme.


%

%
%

\section*{Acknowledgment}
This work was supported by the Ministry of Education and Science of Russian Federation (Project 14.Z50.31.0031), Alstom Grid and partially supported by the iCODE institute, research project of the Idex Paris-Saclay. The first author would like to thank Johannes Schiffer from TUBerlin for his valuable comments and suggestions to improve Section \ref{sec6}.

%



%
\bibliographystyle{plain}        
\bibliography{bibliografia}

\begin{thebibliography}{10}

\bibitem{abbas}
A.M. Abbas and P.W. Lehn.
\newblock {PWM} based {VSC-HVDC} systems -- {A} review.
\newblock In {\em Power Energy Society General Meeting, 2009. PES '09. IEEE},
  pages 1--9, July 2009.

\bibitem{akagi}
H.~Akagi.
\newblock {\em {Instantaneous Power Theory and Applications to Power
  Conditioning}}.
\newblock Wiley, Newark, 2007.

\bibitem{sandberg}
M.~Andreasson, M.~Nazari, D.~V. Dimarogonas, H.~Sandberg, K.~H. Johansson, and
  M.~Ghandhari.
\newblock {Distributed Voltage and Current Control of Multi-Terminal
  High-Voltage Direct Current Transmission Systems}.
\newblock {\em ArXiv e-prints}, November 2013.

\bibitem{bucher}
M.K. Bucher, R.~Wiget, G.~Andersson, and C.M. Franck.
\newblock Multiterminal {HVDC Networks} -- {W}hat is the preferred topology?
\newblock {\em Power Delivery, IEEE Transactions on}, 29(1):406--413, Feb 2014.

\bibitem{BYRISIWIL}
C.I. Byrnes, A.~Isidori, and J.C. Willems.
\newblock Passivity, feedback equivalence, and the global stabilization of
  minimum phase nonlinear systems.
\newblock {\em Automatic Control, IEEE Transactions on}, 36(11):1228--1240, Nov
  1991.

\bibitem{carra}
J.M. Carrasco, L.G. Franquelo, J.T. Bialasiewicz, E.~Galvan, R.C.P. Guisado,
  Ma.A.M. Prats, J.I. Leon, and N.~Moreno-Alfonso.
\newblock Power-electronic systems for the grid integration of renewable energy
  sources: A survey.
\newblock {\em Industrial Electronics, IEEE Transactions on}, 53(4):1002--1016,
  June 2006.

\bibitem{ANDERS}
S.~Chatzivasileiadis, D.~Ernst, and G.~Andersson.
\newblock The global grid.
\newblock {\em CoRR}, abs/1207.4096, 2012.

\bibitem{chen1}
H.~Chen, Z.~Xu, and F.~Zhang.
\newblock Nonlinear control for {VSC} based {HVDC} system.
\newblock In {\em Power Engineering Society General Meeting, 2006. IEEE}, pages
  5 pp.--, 2006.

\bibitem{chen2}
Y.~Chen, J.~Dai, G.~Damm, and F.~Lamnabhi-Lagarrigue.
\newblock Nonlinear control design for a multi-terminal {VSC-HVDC} system.
\newblock In {\em Control Conference (ECC), 2013 European}, pages 3536--3541,
  July 2013.

\bibitem{berteen}
A.~Egea-Alvarez, J.~Beerten, D.~Van~Hertem, and O.~Gomis-Bellmunt.
\newblock Primary and secondary power control of multiterminal {HVDC} grids.
\newblock In {\em AC and DC Power Transmission (ACDC 2012), 10th IET
  International Conference on}, pages 1--6, Dec 2012.

\bibitem{ESCVANORT}
G.~Escobar, A.J. Van~Der Schaft, and R.~Ortega.
\newblock A {H}amiltonian viewpoint in the modeling of switching power
  converters.
\newblock {\em Automatica}, 35(3):445--452, March 1999.

\bibitem{SHAetal}
S.~Fiaz, D.~Zonetti, R.~Ortega, J.M.A. Scherpen, and A.J. van~der Schaft.
\newblock A port-{H}amiltonian approach to power network modeling and analysis.
\newblock {\em European Journal of Control}, 19(6):477 -- 485, 2013.
\newblock Lagrangian and Hamiltonian Methods for Modelling and Control.

\bibitem{agelidis}
N.~Flourentzou, V.G. Agelidis, and G.D. Demetriades.
\newblock {VSC}-based {HVDC} power transmission systems: An overview.
\newblock {\em Power Electronics, IEEE Transactions on}, 24(3):592--602, Mar
  2009.

\bibitem{FRAZAM}
B.A. Francis and G.~Zames.
\newblock On ${H}^{\infty}$-optimal sensitivity theory for {SISO} feedback
  systems.
\newblock {\em Automatic Control, IEEE Transactions on}, 29(1):9--16, Jan 1984.

\bibitem{gomis}
O.~Gomis-Bellmunt, J.~Liang, J.~Ekanayake, R.~King, and N.~Jenkins.
\newblock Topologies of multiterminal {HVDC-VSC} transmission for large
  offshore wind farms.
\newblock {\em Electric Power Systems Research}, 81(2):271 -- 281, 2011.

\bibitem{HERetal}
M.~Hernandez-Gomez, R.~Ortega, F.~Lamnabhi-Lagarrigue, and G.~Escobar.
\newblock Adaptive {PI} stabilization of switched power converters.
\newblock {\em Control Systems Technology, IEEE Transactions on},
  18(3):688--698, 2010.

\bibitem{isidori}
A.~Isidori.
\newblock {\em Nonlinear Control Systems}.
\newblock Springer-Verlag New York, Inc., Secaucus, NJ, USA, 3rd edition, 1995.

\bibitem{jager}
A.~Jager-Waldau.
\newblock Photovoltaics and renewable energies in {E}urope.
\newblock {\em Renewable and Sustainable Energy Reviews}, 11(7):1414 -- 1437,
  2007.

\bibitem{JAYetal}
B.~Jayawardhana, R.~Ortega, E.~Garcia-Canseco, and F.~F.~Casta\~{n}os.
\newblock Passivity of nonlinear incremental systems: Application to {PI}
  stabilization of nonlinear {RLC} circuits.
\newblock In {\em Decision and Control, 2006 45th IEEE Conference on}, pages
  3808--3812, 2006.

\bibitem{joha}
S.G. Johansson, G.~Asplund, E.~Jansson, and R.~Rudervall.
\newblock Power system stability benefits with {VSC} {DC}-transmission systems.
\newblock In {\em CIGRE Conference, Paris, France}, 2004.

\bibitem{kazm}
M.P. Kazmierkowski, R.~Krishnan, F.~Blaabjerg, and J.D. Irwin.
\newblock {\em Control in Power Electronics: Selected Problems}.
\newblock Academic Press Series in Engineering. Elsevier Science, 2002.

\bibitem{kirby}
NM~Kirby, MJ~Luckett, L~Xu, and Werner Siepmann.
\newblock {HVDC} transmission for large offshore windfarms.
\newblock In {\em AC-DC Power Transmission, 2001. Seventh International
  Conference on (Conf. Publ. No. 485)}, pages 162--168. IET, 2001.

\bibitem{lee}
T.~Lee.
\newblock Input-output linearization and zero-dynamics control of three-phase
  {AC/DC} voltage-source converters.
\newblock {\em Power Electronics, {IEEE} Transactions on}, 18(1):11--22, Jan
  2003.

\bibitem{lund}
H.~Lund.
\newblock Large-scale integration of wind power into different energy systems.
\newblock {\em Energy}, 30(13):2402--2412, 2005.

\bibitem{perez}
M.~Perez, R.~Ortega, and J.R. Espinoza.
\newblock Passivity-based {PI} control of switched power converters.
\newblock {\em Control Systems Technology, IEEE Transactions on},
  12(6):881--890, 2004.

\bibitem{pinto}
R.T. Pinto, S.F. Rodrigues, P.~Bauer, and J.~Pierik.
\newblock Comparison of direct voltage control methods of multi-terminal dc
  ({MTDC}) networks through modular dynamic models.
\newblock In {\em Power Electronics and Applications ({EPE} 2011), Proceedings
  of the 2011-14th European Conference on}, pages 1--10, Aug 2011.

\bibitem{QIUDAV}
L.~Qiu and E.~J. Davison.
\newblock Performance limitations of non-minimum phase systems in the
  servomechanism problem.
\newblock pages 337--349, 1993.

\bibitem{sanchez}
S.~Sanchez, R.~Ortega, R.~Gri\ {n}o, G.~Bergna, and M.~Molinas-Cabrera.
\newblock Conditions for existence of equilibrium points of systems with
  constant power loads.
\newblock In {\em Decision and Control, 2013 52nd IEEE Conference on, Firenze,
  Italy}, 2013.

\bibitem{SANVER}
S.R. Sanders and G.C. Verghese.
\newblock Lyapunov-based control for switched power converters.
\newblock {\em Power Electronics, IEEE Transactions on}, 7(1):17--24, Jan 1992.

\bibitem{SERetal}
M.M. Seron, J.H. Braslavsky, and G.C. Goodwin.
\newblock {\em Fundamental Limitations in Filtering and Control}.
\newblock Springer Publishing Company, Incorporated, 1st edition, 2011.

\bibitem{sun}
S.~Shah, R.~Hassan, and J.~Sun.
\newblock {HVDC} transmission system architectures and control - {A} review.
\newblock In {\em Control and Modeling for Power Electronics (COMPEL), 2013
  IEEE 14th Workshop on}, pages 1--8, June 2013.

\bibitem{shuai}
D.~Shuai and X.~Zhang.
\newblock Input-output linearization and stabilization analysis of internal
  dynamics of three-phase {AC/DC} voltage-source converters.
\newblock In {\em Electrical Machines and Systems ({ICEMS}), 2010 International
  Conference on}, pages 329--333, Oct 2010.

\bibitem{bencha}
J.-L. Thomas, S.~Poullain, and A.~Benchaib.
\newblock Analysis of a robust {DC}-bus voltage control system for a {VSC}
  transmission scheme.
\newblock In {\em AC-DC Power Transmission, 2001. Seventh International
  Conference on (Conf. Publ. No. 485)}, pages 119--124, Nov 2001.

\bibitem{torres}
M.A. Torres.
\newblock Estudio de un sistema {VSC-HVDC} y aplicaci\'{o}n de m\'{e}todo de
  control basado en pasividad.
\newblock Master's thesis, Universidad de Concepci\'{o}n, Chile, March 2007.

\bibitem{VAN}
A.J. van~der Schaft.
\newblock {\em $\mathcal{L}_2$-gain and passivity techniques in nonlinear
  control}.
\newblock Communications and control engineering. Springer, Berlin, 2000.

\bibitem{avdspartial}
A.J. van~der Schaft.
\newblock Characterization and partial synthesis of the behavior of resistive
  circuits at their terminals.
\newblock {\em Systems and Control Letters}, 59(7):423 -- 428, 2010.

\bibitem{jeltsema}
A.J. van~der Schaft and D.~Jeltsema.
\newblock {\em Port-{H}amiltonian Systems Theory: An Introductory Overview}.
\newblock now publishers Inc, 2014.

\bibitem{iravani}
A.~Yazdani and R.~Iravani.
\newblock {\em Voltage--Sourced Controlled Power Converters -- Modeling,
  Control and Applications}.
\newblock Wiley IEEE, 2010.

\bibitem{dorfler}
Dorfler~F. Zhao, J.
\newblock Distributed control, load sharing, and dispatch in {DC} microgrids.
\newblock In {\em American Control Conference}, 2015.

\bibitem{ZON}
D.~Zonetti.
\newblock A port--{H}amiltonian approach to power systems modeling.
\newblock Master's thesis, Universit\' {e} de Paris-Sud XI, June 2011.

\bibitem{ECCZonetti}
D.~Zonetti, R.~Ortega, and A.~Benchaib.
\newblock A globally asymptotically stable dencentralized {PI} controller for
  multi--terminal high--voltage {DC} transmission systems.
\newblock In {\em Proceedings of the 13th European Control Conference on}, Jun
  2014.

\end{thebibliography}

\end{document}